%% file: OpaqueCoverSubmit.tex
\documentclass[11pt]{article}
\usepackage{graphicx,color,epsf}

\newcommand{\tr}[1]{ \begin{list}{}{\setlength{\leftmargin}{#1em}} \item}
\newcommand{\tn}{ \item}
\newcommand{\tl}{ \end{list}}
\newcommand{\bc}{\begin{center} \large \bf}
\newcommand{\ec}{\end{center}}

\newenvironment{proof}%
           {\par\noindent
           {\bf Proof\ \ }}%
           {\raisebox{.35em}{\fbox{\rule{0em}{0em}}}}

\newtheorem{Thm}{Theorem}[section]
\newtheorem{Lem}[Thm]{Lemma}
\newtheorem{Prop}[Thm]{Proposition}
\newtheorem{Cor}[Thm]{Corollary}
\newtheorem{Conj}[Thm]{Conjecture}
\newtheorem{Def}[Thm]{Definition}

\newcommand{\icom}[1]{{}}
\newcommand{\B}{{\cal B}}

\def\blackslug{\hbox{\kern1pt\vrule height6pt width4pt  depth1pt\kern1pt}}

\def\edp{\penalty 500\hbox{\quad\blackslug}\ifmmode\else\par
    \vskip4.5pt plus3pt minus2pt\fi}

\def\qed{\penalty 500\hbox{\quad\blackslug}\ifmmode\else\par
    \vskip4.5pt plus3pt minus2pt\fi}

\title{Minimum Opaque Covers for Polygonal Regions}
\author{J. Scott Provan \and Marcus Brazil \and Doreen A. Thomas \and Jia F. Weng}

\date{\today}

\begin{document}

\maketitle
{}\let\thefootnote\relax\footnote{
{}\\
J.S.~Provan\\
Department of Statistics and Operations Research, University of North Carolina, Chapel Hill, NC 27599-3260 USA\\
E-mail: Scott\_Provan@UNC.edu\\
{}\\
M.~Brazil (Corresponding Author)\\
Department of Electrical and Electronic Engineering, The University of Melbourne, VIC 3010, Australia\\
E-mail: brazil@unimelb.edu.au\\
{}\\
D.A.~Thomas, J.F.~Weng\\
Department of Mechanical Engineering, The University of Melbourne, VIC 3010, Australia
}

\centerline{\bf Abstract} {\small The {\em Opaque Cover Problem}\/ (OCP), also known as the {\em Beam Detector Problem}, is the problem of finding, for a set $S$ in Euclidean space, the minimum-length set $F$ which intersects every straight line passing through $S$.  In spite of its simplicity, the problem remains remarkably intractable.  The aim of this paper is to establish a framework and fundamental results for minimum opaque covers where $S$ is a polygonal region in two-dimensional space. We begin by giving some general results about opaque covers, and describe the close connection that the OCP has with the Point Goalie Problem.  We then consider properties of graphical solutions to the OCP when $S$ is a convex polygonal region in the plane.}

\section{Introduction}

Opaque covers, also known as beam detectors, were studied as far back as 1916 by Mazurkiewicz \cite{Mazurkiewicz}, where they were used by astronomers to describe properties of nebulae.  For a century they have continued to be a topic of interest and intrigue.  This paper considers the two-dimensional form of the problem.  Given a set $S$ in the plane, an {\em opaque cover (OC)} for $S$ is any set $F$ having the property that any line in the plane intersecting $S$  also intersects $F$.  The problem of finding an opaque cover of minimum length for any given planar set $S$ is known as the \emph{Opaque Cover Problem (OPC)}. Intuitively, an opaque cover forms a barrier that makes it impossible to see through $S$ from any vantage point; or equivalently, in the case of the term {\em beam detector}, it can detect any beam of light passing through $S$.  Various restrictions can be placed on $F$ regarding its form, connectivity, and relationship to $S$.  The problem of finding an opaque cover of minimum length appears to have been first studied by Croft in 1969 \cite{Croft}.  An excellent discussion of the problem, together with examples and observations, appears in \cite{Faber}, and an extensive bibliography of related papers can be found in \cite{Dumitrescu}. The problem has also been popularlised in a recent book by Stewart \cite{Stewart}.

In 1987, V.~Akman \cite{Akman} attempted to solve the problem by giving a characterization of an OCP solution along with an associated polynomial-time algorithm to find such a solution.  That paper, however, has been shown to be generally flawed.  In particular, Shermer \cite{Shermer} gave a counterexample to the general structure Akman purported to characterize OCP solutions, and subsequently conjectured his own characterization of OCP solutions.  As this paper shows, that conjecture is also false.  In spite of the fact that all conjectured OCP solutions are piecewise-linear, it is not known whether optimal opaque covers are even measurable, let alone piecewise-linear, and the problem has remained elusive and remarkably intractable.

The aim of this paper is to establish a fundamental framework  for research into OCPs, particularly for the case where the input set $S$ is a polygonal region.  This paper establishes a basic theory of minimum opaque covers, including a number of local optimality conditions.  We use these conditions to construct examples of OCP solutions with previously undiscovered structure that are currently the smallest length solutions for their regions.  At the same time, we identify some of the key open research questions that need to be addressed before further substantial progress can be made on this problem. It is our hope that this paper will provide a useful basis for future work in this area.

The organisation of this paper is as follows. In Section~2 we set up a formal mathematical framework for the study of OCs, and establish some fundamental properties. The section contains a key conjecture, The Graphical Conjecture, which states that a minimum OC for a polygonal region is composed of a finite number of straight-line segments; in light of this conjecture we assume throughout the remainder of the paper that all minimum OCs have this property. In Section~3 we describe the close connection that the OCP has with the Point Goalie Problem. Section~4 is a study of structural and geometric properties of minimum connected OCs. The results in this section are then applied to understanding properties of each connected component in solutions to the general OCP, studied in Section~5. Section~6 concludes the paper and describes some key open questions in this area.

\section{Basic properties of opaque covers}

Let $S$ be a compact connected set in $\Re^2$. An {\em opaque cover (OC)} for $S$ is defined to be a set $F$ in $\Re^2$ having the property that every line in the plane intersecting $S$  also intersects $F$. Informally, the Opaque Cover Problem (OCP) aims to find an opaque cover $F$ for $S$ such that $F$ has the minimum possible ``length''.  This term itself needs to be made precise, since we no nothing about the nature of the set $F$. A suitable choice of measure, proposed in \cite{Faber84} and elsewhere, is the following.

\begin{Def} For set $F\in \Re^2$, the {\em 1-dimensional Hausdorff measure}\/ of $S$ is defined by
\begin{equation}
\label{Hausdorff}
\lambda_1(F)= \lim_{\delta\rightarrow 0}\left( \inf \left\{\sum_{i=1}^\infty \mbox{diam}(E_i) \mid \bigcup_{i=1}^\infty E_i = F, \mbox{diam}(E_i) \leq \delta, i=1,2,\ldots  \right\} \right)
\end{equation}
where $\mbox{diam}(E)$ is the supremum of the distance between any two points of $E$.
\end{Def}

Note that this matches the normal definition of Euclidean length when it is defined, but exists for any set in $\Re^2$.

\begin{Def} The {\em Opaque Cover Problem }\/ (OCP).\\
\textbf{Given:} a compact connected set in $S$ in $\Re^2$.\\
\textbf{Find:} a set $F$ of minimum 1-dimensional Hausdorff measure, such that $F$ is an OC for $S$.
\end{Def}

There are a number of versions of the OCP in the literature that depend upon restrictions on the structure of $F$ as an OC:
\tr 4
[\bf interior OCs:]  $F$ is required to lie entirely inside $S$.
\tn
[\bf rectifiable OCs:] $F$ is required to be composed of a finite number of rectifiable curves.
\tn
[\bf graphical OCs:] $F$ is required to be composed of a finite number of straight-line segments.
\tn
[\bf connected OCs:]  $F$ is required to be graphical and connected.
\tn
[\bf single-path OCs:]  $F$ is required to be graphical and a single path.
\tl

Note that single-path $\Rightarrow$ connected $\Rightarrow$ graphical $\Rightarrow$ rectifiable, and that any of these four restrictions can be combined with the  interior restriction. For clarity, we sometimes refer to a type of OC as \emph{general} if it is not restricted to being interior.

In spite of the apparent generality of some of these restrictions, {\em none}\/ of these properties is known to hold for general unrestricted minimum OCs for {\em any}\/ set $S$ except in the following case, where the proof is obvious.

\begin{Lem}\label{lem1D}
If $S$ is a straight line segment, then the minimum OC for $S$ is the set $S$ itself.
\end{Lem}

There are three more results that hold for general sets $S$ that we mention here.

\medskip

\noindent \textbf{Notation.} For any set $S$ let $\bar S$ represent the convex hull of $S$.

\medskip

\begin{Lem}\label{convex}
A set $F$ is an OC for $S$ if and only if it is an OC for $\bar{S}$.
\end{Lem}

\begin{proof}
An OC for $\bar{S}$ is an OC of $S$ since
$S\subseteq \bar{S}$. For the converse, observe that if $p\in
\bar{S} \setminus S$ then any line through $p$ also passes
through $S$, since $S$ is connected. Hence, every OC of $S$ is
an OC of $\bar{S}$.
\end{proof}

By Lemma~\ref{convex} we can assume without loss of generality for general OCs that $S$ is convex. It is a nontrivial restriction on $S$ for interior OCs.

The proof of the following lemma is obvious.

\begin{Lem}\label{SinFbar}
If $F$ is an OC for $S$ then $S \subseteq \bar{F}$.
\end{Lem}

\begin{Lem}\label{lembounds}
Any OC $F$ for $S$ has length $\lambda_1(F) \geq {\mbox diam}(S)$.
\end{Lem}

\begin{proof}
Let $E_1, E_2,\ldots$ be any collection of sets where the length of the union, $F$, is measured as indicated by Equation (\ref{Hausdorff}). Let $u$ and $v$ be
two points defining a diameter of $S$, and consider sweeping a line $l(x)$
though $S$ that is perpendicular to $uv$ at each point $x$ on $uv$. Each $E_i$ must
intersect all $l(x)$ with $x \in T_i$, where $T_i$ is the projection of $E_i$ onto $uv$. Then
$\mbox{\rm diam}(T_i) \leq \mbox{\rm diam}(E_i)$, and so
$$\sum_{i=1}^\infty
\mbox{\rm diam}(E_i) \geq
\sum_{i=1}^\infty
\mbox{\rm diam}(T_i) \geq |uv| = \mbox{\rm diam}(S).$$
and the inequality follows.
\end{proof}

At present, this seems to be the extent of results about OCPs that can be made without placing further assumptions on the structure of the problem.

Our key assumption involves restricting the nature of the set $S$; we will assume throughout this paper that the boundary  of $S$, $\partial S$, is   a {\em convex polygon}, with vertices denoted by $V_S$.  In view of Lemma~\ref{convex} this also allows us to work with  general polygonally-bounded regions in the case where $F$ is not required to be an interior OC, by replacing $S$ with the (polygonally-bounded) set $\bar S$.  The assumption that $\partial S$ is   a {\em convex polygon} also makes the following important structural conjecture plausible.
\begin{Conj} {\bf (The Graphical Conjecture)} A minimum OC is always graphical.
\end{Conj}
A partial justification for this assumption is a result of Dumitrescu et al.~\cite{Dumitrescu}
\begin{Lem}(\cite{Dumitrescu}, Lemma 1)
Any rectifiable OC for $S$ can be arbitrarily closely approximated by a graphical OC.
\end{Lem}
Graphical OCs have a couple of nice properties.  First, the Hausdorff measure becomes simply the sum of the Euclidean length of the line segments.  Second, we can consider graphical OCs as graphs, and so we can use standard graphical terminology such as vertices, edges, and trees to describe features of the OC.  At the present time there is no polygonally-bounded region $S$ for which a conjectured solution to OCP on $S$ has more than $|V_S|$ straight-line segments.  It is interesting to note that Bagemihl in \cite{Bagemihl} has constructed an instance of an OC for a rectangle that is {\em nowhere}\/ rectilinear, although it has infinite Hausdorff measure and is hence clearly not minimal.  Whether the Graphical Conjecture holds remains one of the major open questions for the OCP.

If we consider the graphical OCP, we can strengthen the bounds of Lemma~\ref{lembounds}.  For a finite set $X$ of points in  $\Re^2$, a {\em minimum Steiner tree (MST)}\/ for $X$ is the minimum length set (under the measure in Definition~2.1) connecting  all points in $X$. Such a set is a tree which may contain nodes other than the points in $X$. The fundamental properties of MSTs  are well understood (see, for example, \cite{Hwang}).  Let $\rho(S)$ denote the length of $\partial S$, and let st$(S)$ be the length of an MST on $V_S$. The following result  has appeared in a number of places \cite{Dumitrescu,Faber84,thurston}

\begin{Lem}\label{lembounds2}
Let $F$ be a minimum graphical OC for $S$ in $\Re^2$.  Then
$$\rho(S)/2\leq \lambda_1(F)\leq \mbox{\rm st}(V_S).$$
\end{Lem}

The final lemma of this section indicates the close relationship between OCP solutions and MSTs.
\begin{Lem}\label{steiner}
Let $F$ be a minimum graphical OC for $S$.  Then each component $C$ of $F$ is an MST on $V_{\bar C}$.
\end{Lem}
\begin{proof}
Any line that is blocked by $C$ will be blocked by any set connecting $V_{\bar C}$.  Thus in order that $F$ have minimum length, the components will have to be MSTs.
\end{proof}

\section{Opaque covers and the Point Goalie Problem}

An interesting related problem to OCP has been studied by Richardson and Shepp \cite{Richardson}, and provides a good discrete analogue to OCP.  For $x\in \Re^2$ and $\delta>0$ let $B_\delta(x)$ denote the closed disk of diameter $\delta$ centered at $x$, and for subset $X$ of $\Re^2$ let $\B_\delta(X)=\cup_{x\in X}B_\delta(x)$.  A {\em $\delta$-point goalie}\/ ($\delta$-PG) for convex polygonally-bounded set $S$ in $\Re^2$ is a any set $X$ for which $\B_\delta(X)$ blocks every straight line through $S$.   The intuition for the name is that we wish to block straight-line hockey shots through $S$ by positioning a number of ``goalies'' on the plane each of whom can protect a disk of diameter $\delta$.  The {\em $\delta$-point goalie problem}\/ is to find the minimum cardinality of an $\delta$-PG for $S$. Note that since $S$ is compact then there is always a finite $\delta$-point goalie for any $\delta$.

In \cite{Richardson} the authors consider a form of ``LP-relaxation'' for point goalies.  They compare it to a certain measure on $\Re^2$, called a {\em tomographic goalie}, that turns out to be a limiting measure for ``LP-goalies''.
They also briefly mention minimum OCs, giving them the name ``minimum length goalies'', although they do not consider the relationship between minimum length goalies and $\delta$-goalies.  We investigate limiting results for $\delta$-goalies by asking what happens to $\delta$-PG's when $\delta$ approaches zero.  Define the {\em limiting point goalie value}\/ for $S$ to be
\begin{equation}\label{P(S)}
 P(S)=\displaystyle \lim_{\delta\rightarrow 0}  \min\{\delta|X_\delta|:\mbox{$X_\delta$ is an $\delta$-PG for $S$}\}.
\end{equation}
The value $\delta|X_\delta|$ is simply the sum of the diameters of each of the disks in $\B(X_\delta)$, and so $P(S)$ corresponds to a discrete approximation of the Hausdorff measure of $\B(X_\delta)$, with the $E_i$ sets represented by the disks.

In general it is not clear that the limit for $P(S)$ exists, or what relationship it has with OCP solutions.  Under the Graphical Conjecture, however, the limiting point goalie value $P(S)$ does exist, and is in fact the same as that of minimum OC in the following sense.  Let $\{X_\delta:\,\delta>0\}$ be a set of $\delta$-PGs, and let $F$ be a finite set of lines segments in $\Re^2$.  Define
\[
X^c_\delta(F) = \{x\in X_\delta:\ F\cap B_\delta(x)=\emptyset\}.
\]
We say that $F$ is a {\em limit}\/ of $\{X_\delta\}$ if
\[\begin{array}{ll}
(L1)&\displaystyle \lim_{\delta\rightarrow 0} \lambda_1\left(F\setminus \B_\delta(X_\delta) \right)=0,\mbox{ and}\\[.8em]
(L2)&\displaystyle \lim_{\delta\rightarrow 0} \delta|X^c_\delta(F)| = 0
\end{array}\]
That is, the appropriate measure of the ``symmetric difference'' between the $\delta$-PGs and the set $F$ goes to zero.
\begin{Thm}
Let the Graphical Conjecture hold, let $F\subset\Re^2$ be a finite set of line segments, and let $\{X_\delta:\,\delta>0\}$ be a set of min-cardinality $\delta$-PGs where the limit is $F$.  Then:
\tr 2
[$(i)$] $F$ is a solution to OCP;
\tn
[$(ii)$] the limit in  Equation\/ {\rm (\ref{P(S)})} exists; and
\tn
[$(iii)$] $\displaystyle P(S)=\lambda_1(F).$
\tl
\end{Thm}
\begin{proof}
Let $E_1,\ldots,E_k$ be the set of line segments of $F$, with $E_i$ of length $l_i$, i=1,\dots,k.  First, suppose that $F$ is not an opaque cover, that is, there is some line $L$ that intersects $S$ but not $F$.  Since $S$ and $F$ are both compact, this means that there is a strip $W$ of width $\Delta>0$ around $L$ such that each of the parallel lines in $W$ intersects $S$ but misses $F$.  Therefore it will be necessary to have at least $\lfloor\Delta/\delta\rfloor$ disks $B_\delta(x)$, $x\in X_\delta$ that do not intersect with $F$ in order to cover this set of parallel lines Thus for each $\delta$ there are at least $\lfloor\Delta/\delta\rfloor$ points in $X^c_\delta(F)$, and so the collection $\{X_\delta\}$ does not satisfy $(L2)$.  It follows that no such line $l$ exists, and so $F$ is an OC.

We next construct upper and lower bounds on $\delta|X_\delta|$.   For the upper bound we construct an appropriate $\delta$-PG for $F$.  For each segment $E_i$, place a set $X^i_\delta$ of $\lceil l_i/\delta\rceil$ points of distance $\delta$ from each other and at most $\delta/2$ from the endpoints of $E_i$, so that $E_i\subset \B_\delta(X^i_\delta)$.  If we set $X_\delta=\bigcup_i X^i_\delta$, then each $X_\delta$ is a $\delta$-point goalie of cardinality
\[
\displaystyle |X_\delta| = \sum_{i=1}^k \lceil l_i/\delta\rceil,
\]
and so
\[
\displaystyle \delta|X_\delta|\leq \sum_{i=1}^k \delta\lceil l_i/\delta\rceil.
\]

For the lower bound, let $\{X_\delta\}$ be a collection of $\delta$-PGs converging to $F$.  First note that for sufficiently small $\delta$, any disk $B_\delta(x)$ for which $x$ is distance more than $M\delta$ from the endpoint of any line segment ($M$ being approximately the cosecant of the smallest angle between any two adjacent segments) can cover at most one of the line segments.

Choose an arbitrary $\epsilon>0$.  By $(L1)$ we can choose $\Delta_0$ so that for any $\delta<\Delta_0$ we have
\[
\displaystyle \lambda_1\left(F\setminus \B_\delta(X_\delta) \right)\leq\epsilon/2,
\]
or in other words,
\[
\displaystyle \lambda_1\left(F\cap \B_\delta(X_\delta) \right)>\lambda_1(F)-\epsilon/2.
\]
Now the total length of the portions of the $k$ line segments within $M\delta$ of their endpoints is $2kM\delta$, and the remaining length $\lambda_1(F)-2kM\delta-\epsilon/2$ of $F\cap \B_\delta(X_\delta)$ must be covered by individual disks.  Thus by choosing
\[
\Delta = \min\left\{\Delta_0,\frac{\epsilon}{4kM}\right\}
\]
we have
\[
\delta|X_\delta|\geq \lambda_1(F)-2kM\delta-\epsilon/2\geq \lambda_1(F) - \epsilon.
\]
To summarize, for any $\epsilon$ we have
\[
\lambda_1(F) - \epsilon\leq\delta|X_\delta|\leq \sum_{i=1}^k \delta\lceil l_i/\delta\rceil
\]
holding for all sufficiently small $\delta$.  Since the left-hand and right-hand-sides both converge to $\lambda_1(F)$, this establishes $(ii)$ and $(iii)$.  Further, this also implies that $F$ has minimum length, since by the Graphical Conjecture any smaller-length solution must also be made up of a finite number of line segments, and hence for all sufficiently small $\delta$ we can construct appropriate $\delta$-PGs $X^0_\delta$ as above such that each $|X^0_\delta|$ is bounded above by a value less than $\lambda_1(F)$.  This violates $(iii)$.
\end{proof}

Sadly, there is nothing more known about the general structure of $\delta$-goalies or their limits than there is for OCPs.  We end the section with a conjecture.

\begin{Conj}
Independent of the Graphical Conjecture, a limit of point goalies always exists, and is a minimum OC.
\end{Conj}

\section{Connected opaque covers}

For the remainder of the paper we  continue to assume that $S$ is a convex polygonally-bounded set, and in addition we  consider only graphical OCs. In this section we focus exclusively on connected OCs, both because they are simpler in structure than general graphical OCs, and because the properties of minimum connected OCs give an insight into the structure of the connected components of minimum graphical OCs.

\subsection{Interior connected OCs}

The simplest types of OCs are the interior connected OCs; the interior connected OCP is the only known OC problem for which the solution  can be fully  characterised and computed. The following result is an immediately consequence of Lemma~\ref{steiner}.

\begin{Cor}\label{connected interior OC} A minimum  interior connected OC for $S$ is an MST on $V_S$.
\end{Cor}

This leads to the following result.

\begin{Thm}\label{connected_interior} The interior connected  OCP is NP-hard, but has a fully polynomial approximation scheme.
\end{Thm}
\begin{proof}
The NP-hardness follows from \cite{Brazil}, where it is shown that the MST problem is NP-hard even in cases where the points lie on two parallel straight lines.  These points can be perturbed slightly to form the boundary vertices of a convex set, while still maintaining the NP-hardness properties.  A fully polynomial approximation scheme for the MST problem can be found in \cite{Provan}.
\end{proof}
\\[1em]
We note that, despite the NP-hardness of the MST problem, there an efficient solution software tool, GeoSteiner, for finding general MSTs \cite{WWZ}. This means  that fairly large interior OCPs  (for regions $S$ with hundreds of vertices) can be solved in reasonable time.

Interior single-path OCs are considerably easier to find, since they amount to finding the minimum length path that goes through every point of $V_S$.  This problem is known as the {\em Traveling Salesperson Path Problem (TSPP)}, which is the variant of the Traveling Salesperson problem where the beginning and ending points of the tour do not have to be the same.  The following result was proved in \cite{Dumitrescu}.
\begin{Thm}
TSPP on $V_S$, and hence the interior single-path OCP problem on $S$, can be solved in $O(|V_S|^2)$ time.
\end{Thm}

\subsection{General connected OCs}

It is not known whether a minimum OC can always be chosen to be contained entirely in the set $S$, but this is definitely not true for connected, or even single-path, OCs.  Figure~\ref{connectedOC} illustrates many of the terms and results of this section.  This figure gives three OCs on the same nonagon $BCDEFGHJM$.  The figure has been drawn significantly out-of-scale in order to illustrate its structure more clearly.  Figure~\ref{connectedOC}(a) gives the conjectured minimum connected OC for $S$, consisting of the single path $ABCDLEKFGHI$ of length 4091.17.  We will specify its structure more clearly in the discussion below.  Figure~\ref{connectedOC}(b) is the solution to the interior connected OCP as indicated in Corollary~\ref{connected interior OC}, also a single path $BCDMEJFGH$ of length 4100.58.  Figure~\ref{connectedOC}(c) is a popular conjectured optimal structure (see, for instance, the examples in \cite{Brakke}), called the ``tent'', which involves dropping perpendiculars from the right- and left-most points of $S$ to a horizontal line tangent to $S$ and connecting them with the remaining boundary path enclosing $S$.  In this case the OC is $NBCDEFGHP$, for total length 4093.04.  Thus the tent path, while better than the best interior connected OC, is still not the optimal connected OC for this set.
\begin{figure}
   \begin{center}
       \epsfxsize=6in
       \epsffile{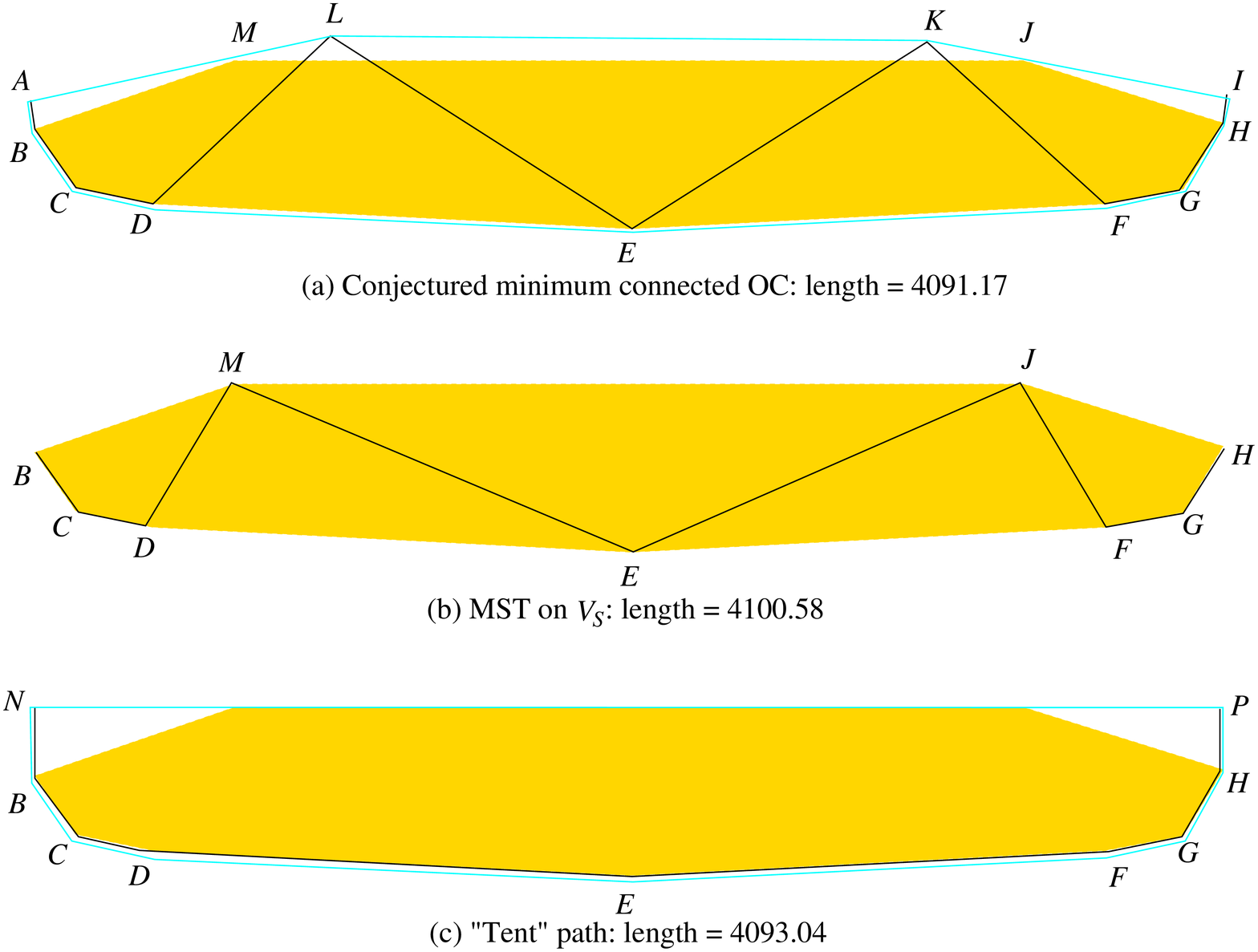}
   \end{center}
\centerline{Coordinate Values}\ \\[-2em]
\[\begin{array}{|c|rr|}
\multicolumn{1}{c}{}&\multicolumn{1}{c}{x}&\multicolumn{1}{c}{y}\\
 \hline
A & -2030.0100 &    35.0000\\
B & -2030.0000 &    33.0000\\
C & -2018.0000 &    15.0000\\
D & -2000.0000 &     5.0000\\
E &     0.0000 &     0.0000\\
F &  2000.0000 &     5.0000\\
G &  2018.0000 &    15.0000\\
H &  2030.0000 &    33.0000\\
I &  2030.0100 &    35.0000\\
J &  1944.1117 &    35.4295\\
L &  1199.3837 &    39.1531\\
M & -1199.3837 &    39.1531\\
K & -1944.1117 &    35.4295\\
\hline
 \end{array}
 \]
   \caption{Three opaque covers}
   \label{connectedOC}
\end{figure}

It follows immediately from Lemma~\ref{steiner} that a minimum connected OC $F$ on $S$ is an MST on $V_{\partial\bar{F}}$, and so a key question is to understand the relationship between $F$, $S$ and $\bar F$.  We begin by examining properties of the set $\bar F$ for a connected OC $F$.

\begin{Lem}\label{conn_hull_lemma}
Let $F$ be a minimum connected OC on $S$.  Then \vspace{-2mm}
\begin{enumerate}\renewcommand{\labelenumi}{$(\alph{enumi})$}
\item
$\bar F$ must contain $S$.
\item \label{conn_hull_lemma2}
The vertices of $\partial \bar F$ are exactly the degree 1 or 2 vertices of $F$.
\item
Every edge of $\partial \bar F$ is either an edge of $F$ or has only its endpoints in $F$.
\item\label{conn_hull_lemma3}
Let ${uv}$ and ${vw}$ be adjacent edges of $\partial \bar F$.
\begin{enumerate}\renewcommand{\labelenumii}{$(\roman{enumii})$}
\item
${uv}\cup {vw}$  intersects $S$ at some point other than $u$ and $w$.
\item
If exactly one of ${uv}$ and ${vw}$, say ${uv}$, is in $F$ then $S$  intersects ${uv}$ at some point other than $u$.
\item
If both ${uv}$ and ${vw}$ are in $F$, then $v$ is in $S$.
\end{enumerate}
\item
There are at most $2|V_S|$ vertices in $\partial\bar{F}$.
\end{enumerate}
\end{Lem}
\begin{proof}
$(a)$
This follows directly from the definition of an OC.

$(b)$
Since $F$ is an MST on the vertices of $\partial \bar F$ the Steiner vertices of $F$ cannot lie on the boundary of $\bar F$ (since the three incident edges at a Steiner point form angles of $2\pi/3$ with each other). Furthermore,  the non-Steiner vertices must all be vertices of $\partial \bar F$, and have degree $1$ or $2$.

$(c)$
Again, since $F$ is an MST on the vertices of $\partial \bar F$, there can be no vertices of $F$ lying between the endpoints of an edge of $\partial \bar F$.

$(d)$
If any of the given conditions is violated, then we can move one of the three  points $u$, $v$, or $w$ inward slightly without violating the OC property to obtain a shorter OC. (This is illustrated in Figure~\ref{connectedOC}(a) by the edges ${BA}$ and ${AL}$ for properties (i) and (ii), and by edges ${AB}$ and ${BC}$ for property (iii).)

$(e)$
From (d) we conclude that the interior of any consecutive pair of edges of $\partial \bar F$ contains at least one element of $V_S$.  Thus there can be no more than $2|V_S|$ edges in $\partial \bar F$ in total.
\end{proof}

Figure~\ref{connectedOC}(a) and (c) are examples of two general connected OCs $F$, both of which satisfy the properties in Lemma~\ref{conn_hull_lemma}.  The polygonal region $\bar F$ for Figure~\ref{connectedOC}(a) has 11 sides, two more than that of $S$, and has one edge $LK$ not touching $S$ at all.  Lemma~\ref{conn_hull_lemma}(d), however, guarantees that the two  edges  adjacent to $LK$ in $\partial \bar F$, namely $AL$ and $KI$,  both intersect $S$ in their interiors, in this case at $M$ and $J$, respectively.

We also note that in Figure~\ref{connectedOC}(a) the edge $MJ$ of $\partial S$ can be replaced by an arbitrary number of boundary edges (keeping $S$ convex) without changing the conjectured OCP solution. Hence, although, by Lemma~\ref{conn_hull_lemma}, at least one in every pair of adjacent edges of $\partial \bar F$ intersects $S$, there can be an arbitrary number of adjacent edges of $\partial S$ that do not intersect $\bar F$.

For the general connected OC problem we can, in theory, find the minimum OC if we can first identify the (at most $2|V_S|$-vertex) region $\bar F$.  In order to understand the properties of $\bar F$, the following definitions will be useful (in this and later sections).

\begin{Def}\label{critical_definition} \textbf{Critical lines, points and edges.} \vspace{-2mm}
\begin{enumerate}
\item A {\em critical point}\/ of a graphical OC $F$ for $S$ is a vertex $v$ of $F$ that is not in $V_S$ but which can be perturbed, along with its adjacent edges in $F$, in such a way that the length of $F$ decreases. Note that if $F$ is a minimum graphical OC then any such perturbation will necessarily result in the existence of some violating line $L$ that intersects $S$ but misses $F$.
\item  A {\em critical line} $L$ of $F$ is the limit of violating lines obtained by length-decreasing perturbations of critical points.
\item An edge of $\partial \bar F$ that is contained in a critical line of $F$ is called a \emph{critical edge} of $F$.
\end{enumerate}
\end{Def}

Note that a critical edge of $F$ is an edge of $\partial \bar F$, but is not necessarily an edge of $F$. In Figure~\ref{connectedOC}(a), the critical points are $A$, $L$, $K$, and $I$, while the critical lines are those passing through the edges $AL$, and $KI$, but {\em not}\/ the one passing through $LK$.  In Figure~\ref{connectedOC}(c) the critical points are $N$ and $P$, and the one critical line is $NP$.

\begin{Lem}\label{critical_points_connected}
Let $F$ be a minimum connected OC for $S$.  Then \vspace{-2mm}
\begin{enumerate}\renewcommand{\labelenumi}{$(\alph{enumi})$}
\item  The critical points of $F$ are precisely the vertices of $\partial \bar F$ that are not vertices of $\partial S$, or equivalently, the degree 1 or 2 vertices of $\partial F$ that are not vertices of $\partial S$.
\item Every critical point is on at least one critical line.
\item Every critical line of $F$ contains a critical edge of $F$.
\item  An edge of $\partial \bar F$ is critical if and only if it is not an edge of $\partial S$ but intersects $S$ at some point other than the endpoints.
\end{enumerate}
\end{Lem}

\begin{proof}
$(a)$  If we move any vertex of $F$ that is not a vertex of $\partial\bar F$ then $F$ will continue to be an OC for $S$.  Thus $F$ could not be made any shorter by doing this.  Conversely, any vertex of $\partial \bar F$ can be moved along its adjacent edge or along the bisector of the angle of its pair of adjacent edges, so as to decrease the length of $F$.  The equivalent characterization follows from Lemma~\ref{conn_hull_lemma}(b).

$(b)$ This follows from the definition of a critical line.

$(c)$ Suppose  a critical line $L$ does not contain an edge of $\partial \bar F$.  Then it must either \emph{(i)} miss $F$ entirely, \emph{(ii)} intersect the interior of $\bar F$, or \emph{(iii)} touch only a vertex $v$ of $\partial\bar F$.  Options \emph{(i)} and \emph{(ii)} cannot occur, since then no perturbation of $F$ and $L$ will cause $L$ to intersect $S$ and miss $F$. Option \emph{(iii)} cannot occur, since if the point $v$ is in $S$ then it cannot be critical, and if $v$ is not in $S$ then there is no way of perturbing $L$ so that it passes through $S$.

$(d)$ Let $e$ be a critical edge of $F$.  If $e$ did not intersect $S$ in its interior, then no noncritical endpoint of $E$ could be perturbed in such a way as to allow a perturbed line to pass through $S$.  Conversely, suppose $e$ is not an edge of $S$, but has a point $u$ of $S$ in its interior.  Then one of its endpoints $v$ is not in $S$, and perturbing $v$ along its adjacent edge or angle bisector must expose $u$ to the line passing through $L$.
\end{proof}

Note that not every edge of $\partial \bar F$ that is not an edge of $S$ will necessarily be a critical edge, as the example in Figure~\ref{connectedOC}(a) illustrates.

\begin{Def}\label{FCP_definition}
A critical point $v$ is a {\em free critical point (FCP)}\/ if it lies on only one critical line.
\end{Def}

In Figure~\ref{connectedOC}, the critical edges $AL$ and $KI$ in (a) and the critical edge $NP$ in (c) all have both of their endpoints free.  The  importance of FCPs is that they provide a method of further restricting solutions of OCPs, since they can be moved along a critical edge in either direction without violating the conditions of $F$ being an OC.  This allows us to place conditions on the positions and angles of the edges of $F$ adjacent to FCPs in a minimum connected OC.

\begin{Thm}\label{boundary-critical_connected}
Let $F$ be a minimum connected OC for $S$, let $E={v_1v}_3$ be a critical edge of $\partial \bar F$, and let $L$ be the line through $E$.  Referring to the labels given in Figure \ref{boundary_cl} we have
\begin{enumerate}\renewcommand{\labelenumi}{$(\alph{enumi})$}
\item
For $i=1,3$, if $v_i$ is an FCP of degree 1 in $F$, then the adjacent edge must be perpendicular to $L$. If $v_i$ is an FCP of degree 2 in $F$, then the two adjacent edges to $v_i$ in $F$ form equal angles $\alpha_i\leq \pi/6$ with $L$.
\tr{-2}
Suppose that both $v_1$ and $v_3$ are FCPs and $L\cap S$ is a single point $v_2$, as in Figure~\ref{boundary_cl}(a). For $i=1,3$, let $d_i$ be the distance $|v_iv_2|$; then $F$ has the following properties:
\tl
\item
If $v_1$ and $v_3$ have degree 1 in $F$, then $d_1=d_3$.
\item
If $v_1$ has degree 1 and $v_3$ has degree 2 in $F$, then $d_1=2d_3\sin\alpha_3$.
\item
If $v_1$ and $v_3$ have degree 2 in $F$, then $d_1\sin\alpha_1=d_3\sin\alpha_3$.
\tr{-2}
Next, suppose $L\cap S$ is an edge ${v_2v_2'}$ of $\partial\bar{S}$, with $v_1,v_2,v_2',v_3$ in order from left to right on $L$, as in Figure~\ref{boundary_cl}(b). Let $d_1=|v_1v_2|$, $d_2=|v_2v_2'|$ and $d_3=|v_2'v_3|$ ; then $F$ has the following properties:
\tl
\item
If both $v_1$ and $v_3$  are of degree 1 then  $d_1\leq d_2+d_3$ and $d_3\leq d_1+d_2$.
\item
If $v_1$ is of degree 1 and $v_3$ is of degree 2, then
\[
\frac{1}{2}\left(\frac{d_1}{d_2+d_3}\right)\leq \sin \alpha_3\leq\frac{1}{2}\left(\frac{d_1+d_2}{d_3}\right).
\]
\item
If both $v_1$ and $v_3$ are of degree 2 then
\[
\frac{d_1}{d_2+d_3}\leq \frac{\sin \alpha_3}{\sin\alpha_1}\leq\frac{d_1+d_2}{d_3}.
\]
\end{enumerate}
\end{Thm}

\begin{figure}[h]
\begin{center}
\input{cl_boundary.tex}\\[1.5em]
(a) $L\cap S$ is a single point\\[1.5em]
\input{cl_boundary_edge.tex}\\[1.5em]
(b) $L\cap S$ is an edge of $\partial \bar S$.
\caption{The two cases for $L$ in Theorem~\ref{boundary-critical_connected}}
\label{boundary_cl}
\end{center}
\end{figure}
\begin{proof}
For $(a)$ we note that in both the degree 1 and 2 case if the angle condition does not hold then we can move $v_i$ slightly in the direction of the smaller of the incident angles and decrease the length of the OC.  The upper bound on $\alpha_i$ comes from the fact that the angle between the edges of $F$ can be no less than $2\pi/3$, since $F$ is an MST.

Statements (b)--(g) all use the following argument from variational calculus. Suppose we have a vertex $v_i$ on $L$ of degree 1 in $F$ where the incident edge in $F$ is orthogonal to $L$. Let $v$ be another point on $L$ at distance $d$ from $v_i$. If we rotate $L$ about  $v$ so that the edge adjacent to $v_i$ increases in length,  then the infinitesimal rate of increase in the length of that edge is $d\sin(\pi/2) = d$. Suppose on the other hand, the vertex $v_i$ on $L$ is of degree 2 in $F$ where the two incident edges both lie on the same side of $L$ and both form equal angles $\alpha_i$ with $L$. If we rotate $L$ about  $v$ so that both edges adjacent to $v_i$ increase in length,  then the infinitesimal rate of increase in the length of each edge  is $d\sin\alpha_i$ and hence the total increase in length is $2d\sin\alpha_i$.
To prove each of the statements (b)--(g), simply note that if any of the (in)equalities do not hold then we can rotate $L$ about $v_2$ or $v'_2$ slightly in the appropriate direction, moving the adjacent edges to $v_1$ and $v_3$ accordingly, so that the resulting set is still an OC and has smaller length than $F$.
\end{proof}

\begin{Cor}\label{boundary-critical_single-path}
Let $F$ be a minimum single-path OC for $S$.  Then Lemmas \ref{conn_hull_lemma} and \ref{critical_points_connected}, and Theorem~\ref{boundary-critical_connected} hold, except that in Theorem~\ref{boundary-critical_connected} the requirement in (a) that $\alpha_i\leq \pi/6$ is no longer required.
\end{Cor}

Referring again to Figure~\ref{connectedOC}, we have
\begin{itemize}
\item The degree one vertices $A$ and $I$ in (a) and $N$ and $P$ in (c) satisfy Theorem~\ref{boundary-critical_connected}(a).
\item
The endpoints of the critical edges $AL$ and $KI$ in (a) satisfy Theorem~\ref{boundary-critical_connected}(c).
\item
The endpoints of the critical edge $NP$ in (c) satisfy Theorem~\ref{boundary-critical_connected}(e).
\end{itemize}

Theorem~\ref{boundary-critical_connected} is limited, covering only the critical edges of $\partial \bar F$ for which the endpoints are FCPs.  The situation can become much more complicated when two or more consecutive edges of $\partial \bar F$ are critical.
In these cases, altering the shape of $\bar F$ may result in a very complex change in the position and length of the MST for $\partial\bar{F}$.  Thus even for general connected or single-path problems the OCP solutions appear to be difficult to characterize, let alone find. We note, however, that we know of no examples of polygonal regions $S$ for which the conjectured minimum connected OC has a critical edge, the endpoints of which are not both FPCs.

We finish the section by solving one special case of the connected OCP.  Namely, when $S$ is a triangular region we show that $\bar F = \bar S$, and hence the minimum connected OC is an MST on $V_S$. For this theorem we first recall some well-known facts about MSTs (see, for example, \cite{Hwang}).

Let $T$ be an MST on point set $V$. Let $|T|$ denote the total length of $T$ (that is, $|T| := \lambda_1(T)$). A {\em Steiner vertex}\/ of $T$ is a vertex of $T$ that is not in $V$.  Then
\begin{enumerate}
\item
$T$ is made up of at most $2|V|-3$ edges.
\item
Each vertex of $T$ has degree 1, 2, or 3.
\item
All Steiner vertices have degree 3, and the incident edges form angles of $2\pi/3$ with each other.
\item
The degree 2 vertices have their adjacent edges forming angles of at least $2\pi/3$.
\end{enumerate}

An MST interconnecting a given set of $n$ points is said to be \emph{full} if each of the given points has degree $1$ in the MST; such a tree has exactly $n-2$ Steiner vertices. Suppose $T$ is a full MST on the set $\{ a,b,c\}$ with Steiner vertex $s$, and let $e=e_c(a,b)$ be the unique point such that $\triangle abe$ is an equilateral triangle and $e$ lies on the opposite side of the line through $ab$ to $c$. Then $ce$ is known as the \emph{Simpson line} for $T$, and has the properties that: (1) $|T| = |ce|$, and (2) $s$ lies on $ce$. This relationship between $|T|$ and the Simpson line no longer applies if $T$ is not full, or more specifically if $\triangle abc$ contains an angle strictly greater than $2\pi/3$.  Also note that the length of any edge of $\triangle abc$ is less than $|T|$.

Before proving the main theorem, we require the following lemma.

\begin{Lem}\label{localmax}
Let $p$ and $q$ be given points in  $\Re^2$, and let $r=r(x)$ be a variable point on a given line (in  $\Re^2$) parameterized by the real variable $x$. Let $T(x)$ be the MST on $\{ p, q, r \}$ with length $l(x):= |T(x)|$. If $l(x_0)$ is a local maximum for $l$ then $T(x_0)$ is a non-full MST such that $r$ has degree $2$ and $\angle prq = 2\pi/3$.
\end{Lem}

\begin{proof}
Let $L$ be the given parameterized line in $\Re^2$. If, for some point $r=r(x)$ on $L$, $r$ is of degree 1 in $T(x)$ then $l(x)$ is not a local maximum, since moving $r$ in one of the two possible  directions along $L$ will increase the length of the Simpson line (if $T(x)$ is full), or the full component of $T(x)$ incident with $r$ if $T(x)$ is not full. Similarly, if, for some point $r(x)$ on $L$, $r$ is of degree 2 in $T(x)$ and the angle at $r$ is strictly greater than $2\pi/3$ then $l(x)$ is not a local maximum since moving in one of the two possible  directions along $L$ takes one outside the ellipse with foci $p$ and $q$ and with $r$ on the boundary.
\end{proof}

\begin{Thm}\label{contriangle}
For any  triangular region $S=\triangle abc$, the minimum connected OC for $S$ is the MST on $\{ a, b, c \}$.
\end{Thm}

\begin{proof}
Let $F$ be a connected OC for $S$; then Lemma~\ref{SinFbar} implies that $\bar F$ is a polygonal region containing $S$. By Lemma~\ref{steiner}, the theorem is equivalent to showing that the MST $T_0$  for $\{ a,b,c \}$ is no longer than the MST $T_{\bar F}$  for the vertices of  ${\partial \bar F}$. The strategy of the proof is to take a copy of $S$ and to try to expand it to find a new MST $T$ with the property that $|T_0|\leq |T|$ and $|T|\leq |T_{\bar F}|$.

First note that we can  assume that $a,b,c$ lie on the boundary of ${\bar F}$, since if any of the points is in the interior of ${\bar F}$ it is easy to see that  $S$ can be expanded in such a way as to increase $|T_0|$.
Furthermore, it follows from Lemma~\ref{localmax} that, by moving the vertices of $S$ along edges of ${\partial \bar F}$ so that $|T_0|$ increases, we can assume that two of the vertices of $S$ (say, $a$ and $b$) are vertices of ${\partial \bar F}$, and that the third vertex of $S$ ($c$) is either a vertex of ${\partial \bar F}$ or a degree $2$ terminal of $T_0$ with an angle of exactly $2\pi/3$, lying on a boundary edge of ${\bar F}$. There are three cases to consider.

\smallskip

\noindent \textsc{Case} 1: Suppose $c$ is a vertex of ${\partial \bar F}$; ie, all three points $a,b,c$ are vertices of ${\partial \bar F}$. In this case, $T_{\bar F}$ interconnects  $a,b,c$, hence by the definition of a MST it follows that $|T_0| \leq |T_{\bar F}|$, as required.

\smallskip

\noindent \textsc{Case} 2: Suppose $c$ is a degree 2 terminal of $T_0$ with an angle of exactly $2\pi/3$ and lying on the boundary edge $v_1 v_2$ of ${\partial \bar F}$ where $v_1$ is also a vertex of $S$ (say $v_1=a$); see Figure~\ref{tricases}.
\begin{figure}[ht]
\begin{center}\includegraphics[width=5.0in]{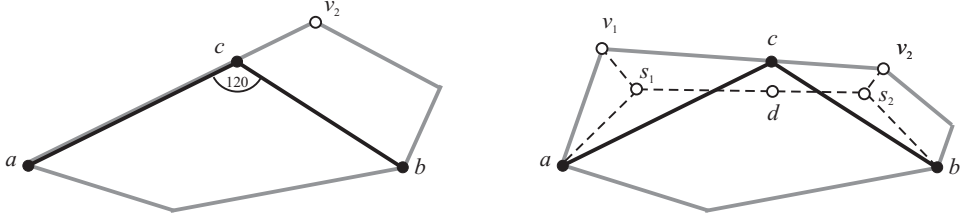}
\end{center}
\caption{Diagrams for Cases~2 and 3 of the proof. The boundary of the polygon ${\bar F}$ is indicated in grey, the MST $T_0$ for $\{ a,b,c \}$ is shown in black. In each case, the angle at $c$ in $T_0$ is  $2\pi/3$.}\label{tricases}
\end{figure}
Let $T_2$ be the MST for $\{ a,b,v_2 \}$. The tree $T_0$ is a degenerate full MST with Simpson line $e_c(a,b)c$. Noting that $\angle e_c(a,b)cv_2 =  2\pi/3 > \pi/2$, it follows that moving $c$ along $v_1 v_2$ towards $v_2$ increases $|T_0|$, and hence $|T_0| < |T_2| \leq |T_{\bar F}|$, where the second inequality again follows from the argument in Case~1.

\smallskip

\noindent \textsc{Case} 3: Suppose $c$ is a degree 2 terminal of $T_0$ with an angle of exactly $2\pi/3$ and lying on the boundary edge $v_1 v_2$ of ${\partial \bar F}$ where  neither $v_1$ nor $v_2 $ is in  $\{ a,b \}$. Let $T_{1,2}$ be a MST for $\{ v_1,v_2,b,a \}$ (where the points are listed in clockwise order around the quadrilateral). By considering the angles between the two diagonals of the quadrilateral $v_1v_2ba$, it can easily be shown, using for example \cite{Du}, that either $T_{1,2}$ has a full topology where one of the Steiner vertices, $s_1$, is adjacent to $a$ and $v_1$ and the other, $s_2$, is adjacent to $v_2$ and $b$, or else  $T_{1,2}$  is a degeneracy of this topology; see the second diagram in Figure~\ref{tricases}. We can assume (by relabeling if necessary) that $|v_1s_1| \geq |v_2s_2| \geq 0$. If $|v_1s_1|=0$ then  $T_{1,2}$ interconnects  $a,b,c$ and hence $|T_0| \leq |T_{1,2}| \leq |T_{\bar F}|$. If $|v_1s_1|>0$ then let $d$ be the closest point on $s_1s_2$ to $c$, and note that $|cd|< |v_1s_1|$. Let $T'_{1,2}$ be the tree obtained from $T_{1,2}$  by removing the edge $v_1s_1$ and adding the edge $cd$. Then $T'_{1,2}$ interconnects $a,b,c$ and hence
$$|T_0| \leq |T'_{1,2}|\leq |T_{1,2}| \leq |T_{\bar F}|$$
giving the required result for the final case.
\end{proof}

\medskip

\textbf{Note:} the above result is mentioned in a few previous papers (such as \cite{Faber}), but this appears to be the first time a proof has appeared in the literature.

\section{Opaque covers with multiple connected components}

For the remainder of this paper we will be considering the general graphical OCP, again for a convex polygonally-bounded set $S$ but now allowing for the possibility of multiple connected components.  We will not be considering the interior version of the problem; although there are a couple of interesting modifications of the material in this section that would apply to the interior problem, they are not very illuminating in their own right.

The conjectured optimal connected solution for the set $S$ given in Figure~\ref{connectedOC}(a) is not the best graphical OC.  A smaller OC for $S$ is based on a construction of Shermer \cite{Shermer}.  He conjectured that the minimum OC for a convex polygonal set $S$ can be obtained by the following procedure:
\begin{enumerate}
\item Triangulate $S$, and then remove all of the diagonals inside some convex region $P_0$ made up of a chosen subset of the triangles.
\item
Find the MST $T_0$ on the vertices of $P_0$.
\item
Recursively for $i=1,\ldots,|V_S|-|V_{P_0}|$ choose a triangle $\Delta_i$  (not in $P_{i-1}$) in the triangulation that shares an edge $e_i$ with $P_{i-1}$, and  let $A_i$  be the altitude of $\Delta_i$ perpendicular to $e_i$.  Set $P_i=P_{i-1}\cup\Delta_i$ and $T_i=T_{i-1} \cup A_i$.
\item
These steps can be repeated for all triangulations of $S$. The final $T_i$ obtained in each case will be an OC of $S$.
\end{enumerate}

In Figure~\ref{disconnectedOC}, $P_0$ is the lower half of $S$, and the perpendiculars are inserted as given in the figure.  In this case Shermer's construction yields what we conjecture to be the minimum OC. However, as we will show below (in the example in Figure~\ref{2C_OC_Figure}),  it does not do so in every instance.

\begin{figure}
   \begin{center}
       \epsfxsize=6in
       \epsffile{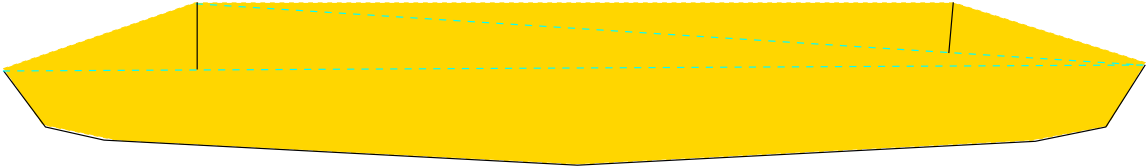}
   \end{center}
   \caption{Conjectured minimum  OC for Figure~\ref{connectedOC}: length = 4089.27}
   \label{disconnectedOC}
\end{figure}

In \cite{Dumitrescu}, Theorem 7, it is shown how to efficiently recognize a graphical set as being an OC for $S$. By observing the configurations obtained from finding violating lines we can give a more precise description of critical points and lines in the multi-component case.

\medskip

\noindent \textbf{Notation:} For a graphical OC $F$ made up of multiple components, let $\cal F$ be the set of extreme points of the convex hulls of the components of $F$.

\begin{Prop}\label{OC_component_lemma}
Let $S$ be a convex polygonally-bounded set, and let $F$ be a finite set of line segments consisting of at least two connected components.  Then $F$ is an OC for $S$ if and only if there do not exist points $u\in \cal F$ and $v\in V_S\cup \cal F$, and a line $L$ intersecting $S$ and containing $u$ and $v$, having the following properties (see Figure~\ref{OC_component_figure}):
\tr 2
[$(a)$] Every component of $F$ lies entirely in one of the half-planes determined by $L$; and
\tn
[$(b)$] In a left-right traversal of $L$, the components of $F$ that intersect $L$ to the left of $v$, and those that intersect $L$ to the right of $v$, lie in opposite half-planes as determined by $L$; and
\tn
[$(c)$] Either $L$ intersects the interior of $S$, or $(i)$ $v\in V_S$, $(ii)$ $S$ and the component of $F$ containing $u$ are both tangent to $L$ from the same side, and $(iii)$ all components tangent to $L$ on the $v$-side of $u$ lie on the opposite side of $L$ to $S$.
\tl
\end{Prop}

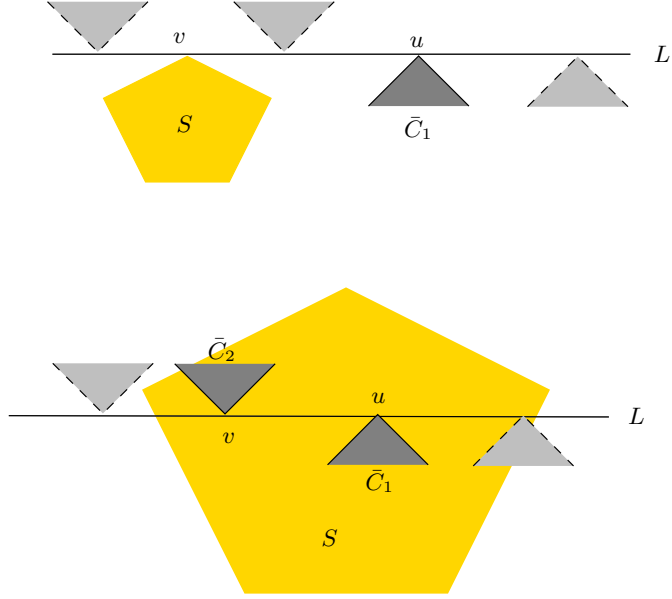
\begin{figure}
\begin{center}
\input{BD_component_figure.tex}\\[1.5em]
\end{center}
\caption{Position of $S$ and components $\bar C$ for Proposition~\ref{OC_component_lemma}.  }
\label{OC_component_figure}
\end{figure}

\begin{proof}
First suppose $F$ is not an OC for $S$.  Then there exists a line $L$ that intersects $S$ but misses $F$, and without loss of generality we may assume that $L$ intersects the interior of $S$ at some point $x$ in general position.  Now rotate $L$ clockwise about $x$ until it first intersects $F$ at a unique point $y$, and let $C_1$ be the component of $F$ containing $y$.   Now continue the rotation by rotating $L$ clockwise around the boundary of $\bar C_1$ until it either becomes tangent to $S$, or is prevented from further rotation by another component $C_2$ of $F$.  Let $u$ be the vertex of $\partial \bar C_1$ representing the final point of rotation of $L$ around $\bar C_1$, and let $v$ be the point of $V_S\cup\cal F$ on $L$ that has prevented $L$ from further rotation.  Note that since $x$ was in general position, these two rotations must be through positive angles.  From the construction it is clear that $(a)$ and $(b)$ are satisfied, and that $L$ continues to intersect $S$.  Further, if $L$ does not now intersect the interior of $S$, then it would be tangent to $S$ at some vertex $v$ of $S$, with $S$ on the same side of $L$ as $\bar C_1$, so that $(c)$ is satisfied.

Conversely, suppose there is a line $L$ and points $u$ and $v$ satisfying $(a)$--$(c)$.  Then a small rotation of $L$ about $u$, followed by a small translation of $L$ away from the component of $F$ containing $u$,  results in a line that intersects $S$ but misses $F$.  Thus $F$ is not an OC for $S$.
\end{proof}

Proposition~\ref{OC_component_lemma} provides a means of verifying whether a given graphical set is an OC for $S$.

\begin{Cor}
It can be determined whether or not a finite set $F$ of line segments is an OC for $S$ in time $O(|{\cal F}|^2|V_S\cup {\cal F}|)$.
\end{Cor}

\begin{proof}
 Using Proposition~\ref{OC_component_lemma}, we can check if a finite set $F$ of line segments is an OC for $S$, by considering every suitable pair $u$ and $v$ and checking whether they satisfy conditions $(a)$--$(c)$ with respect to the line $L$ going through them.  The three conditions simply require checking the location of the elements of $V_S\cup \cal F$ with respect to the line $L$. It is clear that all pairs of points $u, v \in \cal F$ can be checked in time $O(|{\cal F}|^2|V_S\cup {\cal F}|)$. For each pair of the form $u \in \cal F$ and $v \in V_S$, the relative positions of  $L$ and  $S$ can be determined by considering the locations of the two  vertices of $\delta S$ adjacent to  $v$ with respect to  $L$ (due to the convexity of $S$); hence conditions $(a)$--$(c)$ can now be checked by only considering the locations of the elements of $\cal F$. It follows that in this second case all pairs can be checked in time $O(|{\cal F}|^2|V_S|)$. The result follows.
\end{proof}

\subsection{Critical points and critical lines}

For this subsection we assume $F$ is a minimum graphical OC for $S$. From Lemma~\ref{steiner} we know $F$ is a forest for which each component $C$ of $F$ is an MST  on the vertices  of the boundary of its convex hull $\partial V_{\bar C}$.  Further, these components satisfy the conditions of Lemma~\ref{OC_component_lemma}.  Thus an important step in solving the OCP is determining geometric relationships between the polygonal regions determined by these convex hulls.  The key structures for doing this are the critical points, critical lines, and FCPs as given in Definitions \ref{critical_definition} and \ref{FCP_definition}.

The following lemma is obvious.
\begin{Lem}\label{critical_existence_lemma}
The critical points for $F$ are precisely the elements of $\cal F \setminus V_S$.
\end{Lem}

The next result characterizes the orientation of critical points and their adjacent components in relation to their associated critical lines. The statement of the lemma defines three types of critical lines which we will use in later parts of this paper.

\begin{Lem}\label{CL_lemma}  Let $c$ be a critical point for minimum OC $F$. Then up to left/right symmetry, there exists an associated critical line $L$ with points $v_1$, $v_2$, and $v_3$ lying consecutively left-to-right on $L$ and satisfying one of the following three criteria, (see Figure~\ref{CL_lemma_figure}):
\tr 3
[\bf Type I critical line:] $L$ is tangent to $\bar F$.  Then $v_1$ and $v_3$ are distinct elements of $\cal F$, $v_2$ is in $V_S$ (possibly coincident with $v_1$ or $v_3$), $c=v_1$, and no other point of $F$ lies to the left of $v_2$ on $L$.

\tn
[\bf Type II critical line:] $L$ is tangent to $S$ but not tangent to $\bar F$.  Then $v_1$ is in $V_S$ and $v_2$ and $v_3$ are in $\cal F$ with $v_2$ tangent from the same side as $S$ and $v_3$ tangent from the opposite side as $S$, and either
\tr 2
[$(i)$] $c=v_2$ and no other point of $F$ is tangent before $v_3$ on $L$ from the same side as $S$ nor after $v_3$ on $L$ from the opposite side as $S$, or
\tn
[$(ii)$] $c=v_3$ and no point of $F$ is tangent before $v_2$ on $L$ from the same side as $S$ nor after $v_2$ on $L$ from the opposite side as $S$.
\tl
\tn
[\bf Type III critical line:] $L$ intersects the interior of $S$.  Then $v_1$, $v_2$ and $v_3$ are all in $\cal F$, with either
\tr 2
[$(i)$] $c=v_1$ and no other point of $F$ is tangent before $v_3$ on $L$ from the same side as $v_1$ nor after $v_2$ on $L$ from the opposite side as $v_1$,
\tn
[$(ii)$] $c=v_2$ and on at least one side of $v_2$ there are no points of $F$ tangent to $L$ from the same side as $v_2$.
\tl
\tl

\end{Lem}

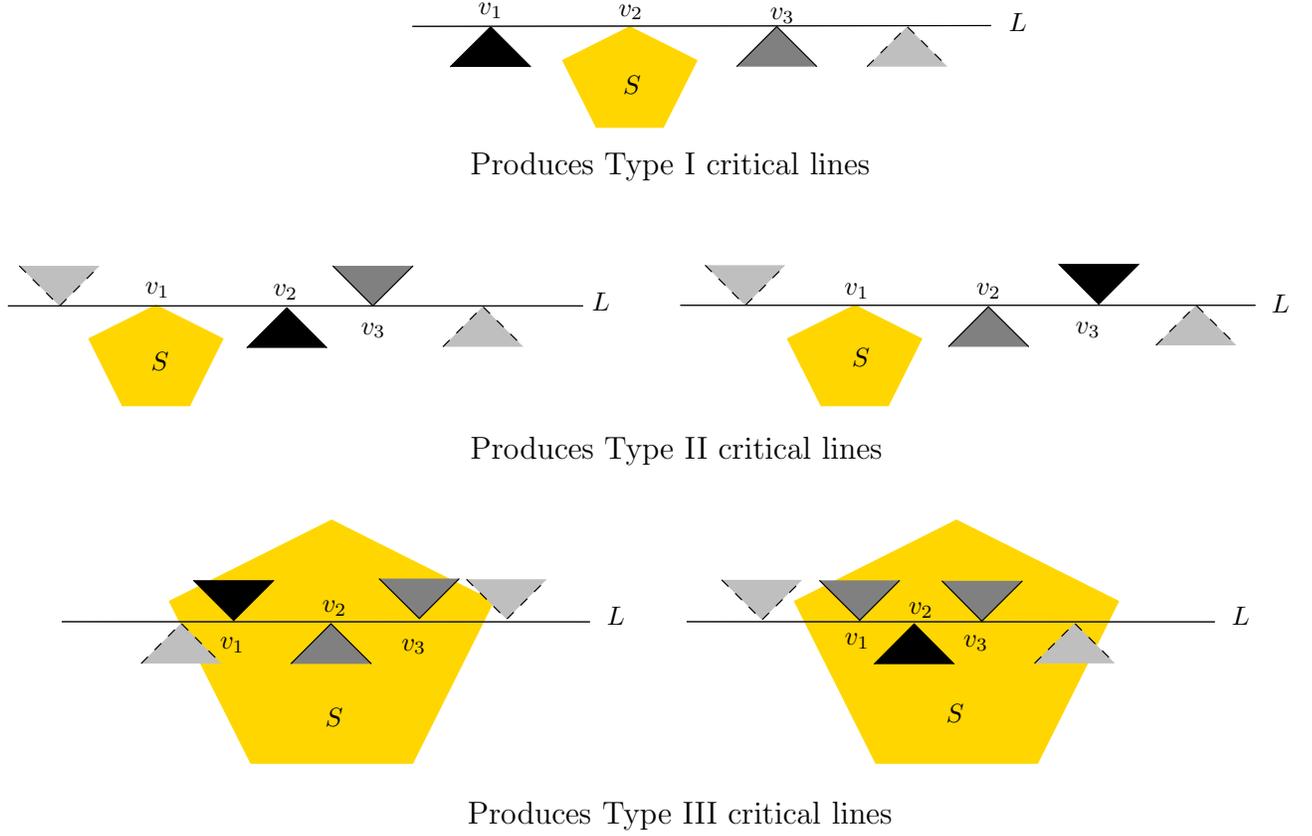
\begin{figure}
\begin{center}
\input{CL_lemma_figure.tex}\\[1.5em]
\end{center}
\caption{Positions of tangent components for Lemma~\ref{CL_lemma}, up to symmetry. Black components are tangent at critical points, grey components are additional required tangent components, and dotted components are optional.}
\label{CL_lemma_figure}
\end{figure}

\begin{proof}
Let $C$ be the component containing $c$.  First perturb $F$ by moving $c$ inward with respect to $\bar C$ in such a way as to decrease the length of $F$, so that the perturbed set $F'$ is no longer an OC. This implies  that there must be a line $L'$ and associated points $u'$ and $v'$ satisfying the three conditions of Proposition~\ref{OC_component_lemma} with respect to  $F'$.  Now consider the nonperturbed points $u$ and $v$ on $F$ corresponding to these, and the line $L_{uv}$ that passes through $u$ and $v$.  Clearly $L_{uv}$ intersects $S$ and does not intersect the interior of the convex hull of any component of $F$, since this is true of the perturbed line $L'$ with respect to $F'$.  Thus, by definition, $L_{uv}$ is a critical line for $c$. Let $L_{uv}$ be the critical line $L$ in the statement of the lemma. We note that $c$ and $\bar C$ and the components of $F$ must prevent the conditions of Proposition~\ref{OC_component_lemma} from occurring. In other words, the components of $F$ tangent to $L$ must prevent any rotation of $L$ that allows $L$ to enter $S$. By considering placement of a single disqualifying tangent component to Figure~\ref{OC_component_figure}, we get the placements given in this lemma.
\end{proof}

\begin{Cor}\label{extra-critical-points}
Let $F$ be a  minimum OC, with a  critical line $L$. Let $v_1$, $v_2$, and $v_3$ be points of $F$ on $L$ satisfying the conditions of Lemma~\ref{CL_lemma}. Then any other vertex of $F$ lying on $L$ must lie on at least two critical lines, and hence is not an FCP.
\end{Cor}

\begin{proof}
This follows immediately from the minimality of $F$. Note that these extra vertices correspond to the dotted optional components indicated in Figure~\ref{CL_lemma_figure}.
\end{proof}

\begin{Cor}\label{non-boundary_critical_line}
Any minimum OC $F$ for $S$ that contains more than
one component must contain at least one critical line of Type II or Type III.
\end{Cor}

\begin{proof}
We begin by proving that if $F$ contains more than one component, then there must be at least one element of $\cal F$ lying in the interior of $\bar F$.  Suppose, on the contrary, that every element of $\cal F$ lies on $\partial \bar F$. Then, since $F$ is a minimum OC every element of $\cal F$ is a vertex of $\partial\bar F$.  Now start at any vertex $v_1$ of $\partial\bar F$, set $i:=1$, and perform the following steps iteratively:
\begin{enumerate}
\item Find the component $C_i$ of $F$ containing $v_i$.
\item If $\bar C_i$ does not intersect the interior of $S$, then let ${xy}$ be the edge of $\partial \bar C_i$ between points $x,y\in\partial\bar F$ separating $\bar C_i$ and the interior of $S$ in $\bar F$.

\item Traverse $\partial\bar F$ from $x$ in the direction of a component of $F$ intersecting the interior of $S$ until the next vertex of $\partial\bar F$,   $v_{i+1}$, is reached.

\item Set $i:= i+1$.
\end{enumerate}
At each iteration we are decreasing the area of $\bar F$ in which we are searching for a component $C_i$ of $F$ for  which $\bar C_i$ intersects the interior of $S$; hence this process must end with a suitable $C_i$.  Note that, for this $C_i$,  $\bar C_i$ does not strictly contain $S$, since otherwise $C_i$ itself would be a one-component OC for $S$.  Thus there must be an edge ${xy}$ of $\bar C_i$, with $x,y\in\partial\bar F$, that intersects the interior of $S$.  But this implies that there must be a component $C'$ such that $\bar C'$ touches ${xy}$ from the opposite side of ${xy}$, since otherwise $F$ would not be an OC.  Since $C'$ cannot contain either $x$ or $y$, then there must be some vertex $c$ of $\partial \bar C'$ lying strictly between them.  Then $c$ lies in the interior of $\bar F$, giving the required contradiction.

By Lemma~\ref{critical_existence_lemma}, $c$ is a critical point of $F$, and hence the result follows  by Lemma~\ref{CL_lemma}.
\end{proof}

\smallskip

\subsection{Further properties of critical lines}

In the case where the critical points on a critical line are FCPs we obtain useful additional properties that further characterize OCP solutions.

We  first note that it is possible to construct examples of OCs for polygonal regions where not all critical points are free. For example, the OC for the triangular region $S$ bounded by $\triangle ABC$ shown in Figure~\ref{nonfree}(a) has three critical points (the vertices of the inner triangle formed by the intersections of the critical lines) none of which is free.
\begin{figure}[ht]
\begin{center}\includegraphics[width=4.5in]{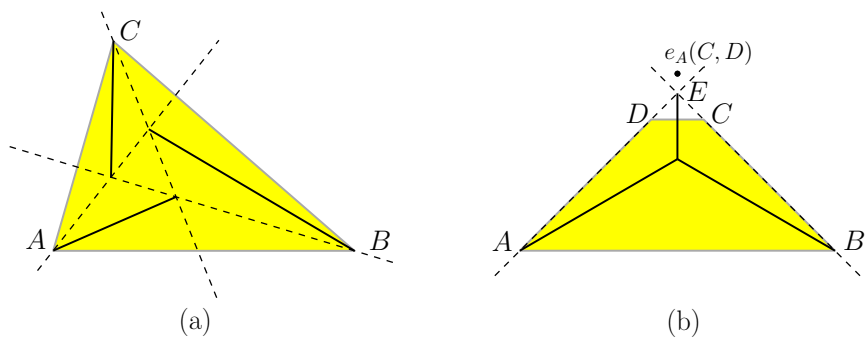}
\end{center}
\caption{Two examples of OC's with critical points that are not free.  Here critical lines are indicated by dashed lines. The example in (a) is not a minimum OC, whereas the example in (b) for the region $ABCD$ is a conjectured minimum OC.}\label{nonfree}
\end{figure}
However, in this example the OC is not minimum; if we move each critical point to the orthocenter of the triangle formed by the three critical points then we obtain a shorter connected OC for $S$ (with no critical points).

Figure~\ref{nonfree}(b) shows an example of a polygonal region $S$ where the conjectured minimum OC contains a critical point that is not free. For this example we assume $ABE$ is a right isosceles triangular region with coordinates $A=(-100, 0), B=(100,0), E=(0,100)$. Given a small positive constant $\varepsilon$, let $S$ be the region $ABCD$ where $C= (\varepsilon, 100-\varepsilon)$ and  $D= (-\varepsilon, 100-\varepsilon)$; intuitively, $S$ is the region $ABE$ with the vertex at $E$ cut off. For $\varepsilon$ sufficiently small the conjectured minimum graphical OC for $S$ is the MST on $\{ A, B, E\}$, as shown in Figure~\ref{nonfree}(b). Using the properties of MSTs outlined in Section~4, it is straightforward to see that this is shorter than the MST on $\{ A, B, C, D\}$ because the equilateral point $e_A(C,D)$ lies outside $\triangle ABE$. For example, if $\varepsilon=1$ then  the length of the MST on $\{ A, B, E\}$ is $273.205$, while the MST on $\{ A, B, C, D\}$  has length $273.937$ (and all other possible OCs appear to be substantially longer). In the conjectured minimum OC, $E$ is a critical point, but is not free, since it lies on two critical lines (as shown in the figure). We note that these two critical lines are both Type~I and are extensions of edges of $S$.

All known examples of polygonal regions where the conjectured minimum OC contains critical points that are not free are similar to the example in Figure~\ref{nonfree}(b), in that each such critical point is determined by two Type~I critical lines that are extensions of edges of $S$. It seems likely (but unproven) that all non-free critical points satisfy this condition. This motivates the study of critical lines where all critical points are FCPs, as it is these points that appear to present the main challenge in constructing minimum OCs.

\begin{Lem}\label{free-pt-properties}  Let $F$ be a minimum  OC for $S$, and let $v$ be an FCP with respect to its associated critical line $L$.  Then
\begin{enumerate}\renewcommand{\labelenumi}{$(\alph{enumi})$}
\item $v$ has
degree 1 or 2 in $F$;
\item if $v$ is of degree 1  then the edge of $F$ incident to $v$ is perpendicular
to $L$;
\item if $v$ is of degree 2 then the two edges
incident to $v$ form an angle with each other of at least $2\pi/3$, and  form equal angles $\alpha \leq \pi/6$ with $L$.
\end{enumerate}
\end{Lem}

\begin{proof} Statement (a) follows from the proof of Lemma~\ref{critical_points_connected}(a). Statements (b) and (c) follow from the proof of Theorem~\ref{boundary-critical_connected}(a). Note that neither of these earlier proofs rely on $F$ being connected.
\end{proof}

For the remainder of this section we focus on critical lines in which all of the incident critical points are free.

\begin{Lem}
Let $L$ be a critical line in which all of the incident critical points are FCPs.  Then in addition to satisfying Lemma~\ref{free-pt-properties} the points $v_1$, $v_2$ and $v_3$ satisfying the conditions of Lemma~\ref{CL_lemma} are the {\em only}\/ points of intersection of $F$ with $L$.
\end{Lem}

\begin{proof} This follows immediately from Corollary~\ref{extra-critical-points}. Note that, for each of the three types of critical line, the components of $F$ tangent to $L$ appear as shown in Figure~\ref{CL_freepoint_figure}.
\end{proof}

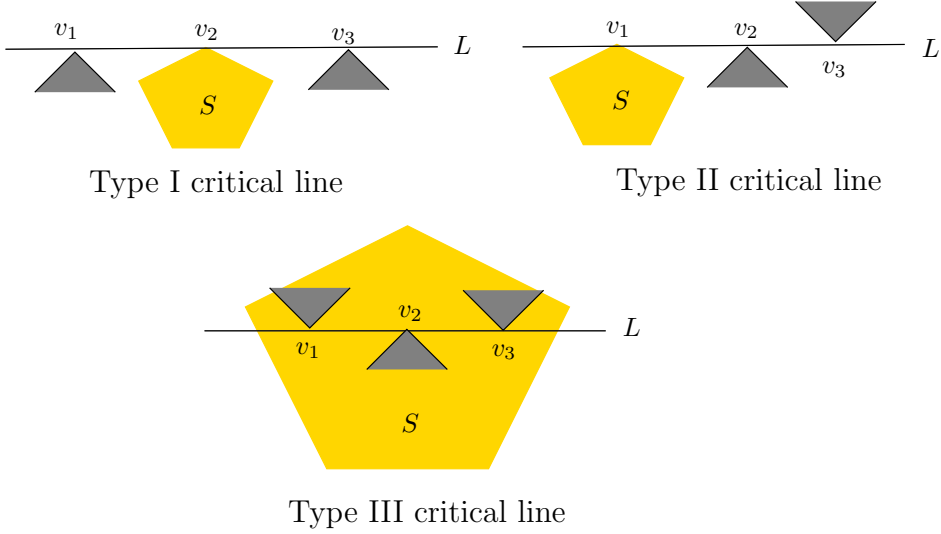
\begin{figure}
\begin{center}
\input{CL_freepoint_figure.tex}\\[1.5em]
\end{center}
\caption{Positions of tangent components of $F$ on $L$ when all critical points of $L$ are FCPs.}
\label{CL_freepoint_figure}
\end{figure}

We will now consider the three types of critical lines separately.

\begin{Thm}\label{TypeI_CL} Let $L$ be a Type I critical line for which all incident points of $\cal F$ are FCPs.  Then $L$ intersects $\bar F$ in a critical edge, the points $v_1$ and $v_3$ are both of degree 2 in $F$, and using the notation in Figure~\ref{boundary_cl} we have:
\tr 2
[(a)] if $L\cap S$ is a single point $v_2$ then $d_1\sin\alpha_1=d_3\sin\alpha_3$ and $\sin\alpha_1 + \sin\alpha_3\leq 1/2$;
\tn
[(b)] if $L\cap S$ is an edge of $S$ then
\[
\sin\alpha_1\leq\frac{1}{2}\left(\frac{d_2 + d_3}{d_1+ d_2 + d_3}\right),\ \
\sin\alpha_3\leq\frac{1}{2}\left(\frac{d_1+d_2}{d_1+ d_2 + d_3}\right);
\]
and
\[
\frac{d_3}{d_1+d_2}\leq
\frac{\sin\alpha_1}{\sin\alpha_3}\leq\frac{d_2+d_3}{d_1}.
\]
\tl
\end{Thm}

\begin{proof}
First suppose $v_1$ has degree 1 in $F$.  Referring to Figure~\ref{boundary_cl}, if $S$ intersects $L$ in one point $v_2$ or  in an edge $v_2v_2'$, then we can rotate $L$ about $v_3$ inwards (with respect to $S$), decreasing the length of the edge adjacent to $v_1$ while growing an edge from $v_2$ perpendicular to $L$. In either case this will decrease the length of $F$ without violating the conditions of an OC, giving a contradiction. A symmetric argument shows that $v_3$ also has degree 2 in $F$.

The equality and inequalities in the statement of the theorem can all be established by showing that they are necessary in order to prevent the existence of rotations of $L$ which decrease the length of $F$.

For $(a)$, we first consider rotations of $L$  about $v_2$ that could potentially decrease the length of $F$.  The infinitesimal rate of change in the length of $F$ from such a rotation  in the direction towards $v_1$'s component is $2d_3\sin\alpha_3-2d_1\sin\alpha_1$, and for the opposite rotation is $2d_1\sin\alpha_1-2d_3\sin\alpha_3$. Hence, both expressions must equal zero, and the equality $d_1\sin\alpha_1=d_3\sin\alpha_3$ follows.  For the inequality, we note that rotating $L$ around, say, $v_1$ towards $v_3$'s component and adding a perpendicular from $L$ to $v_2$, results in an infinitesimal rate of change in the length of $F$ of  $d_1-2(d_1+d_3)\sin\alpha_3$.  Noting that this must be nonnegative, and using the equality above, we obtain the required inequality $\sin\alpha_1 + \sin\alpha_3\leq 1/2$.  Note that the same inequality also results from rotating $L$ about $v_3$.

For $(b)$, we note that a rotation of $L$ about $v_3$ toward $v_1$'s component, adding a perpendicular from $L$ to $v_2$, results in an OC where the infinitesimal rate of change in the length of $F$ is $(d_2+d_3)-2(d_1+d_2+d_3)\sin\alpha_1$.  Since this must be nonnegative, the first inequality follows.  By rotating $L$ about $v_1$ toward $v_3$'s component, we get the second inequality. The last pair of inequalities follows, by a similar argument, from considering rotations of $L$ about $v_2$ and $v_2'$ away from $S$ in each case.
\end{proof}
\begin{Thm}\label{TypeII_CL} Let $L$ be Type II critical line for which all incident points of $\cal F$ are FCPs.  Then, referring to the notation in Figure~\ref{CL_TypeII_figure}, $v_3$ has degree 2 in $F$, and  we have:
\tr 2
[(a)] if $L\cap S$ is a single point and $v_2$ has degree 1, then
\[
\sin\alpha_3=\frac{1}{2}\left(\frac{d_2}{d_2+d_3}\right);
\]

\tn
[(b)] if $L\cap S$ is a single point and $v_2$ has degree 2, then
\[
\frac{\sin\alpha_3}{\sin\alpha_2}=\frac{d_2}{d_2+d_3};
\]
\tn
[(c)] if $L\cap S$ is an edge and $v_2$ has degree 1, then
\[
\frac{1}{2}\left(\frac{d_2}{d_2+d_3}\right)\leq
\sin\alpha_3\leq\frac{1}{2}\left(\frac{d_1+d_2}{d_1+d_2+d_3}\right);
\]

\tn
[(d)] if $L\cap S$ is an edge and $v_2$ has degree 2, then
\[
\frac{d_2}{d_2+d_3}\leq
\frac{\sin\alpha_3}{\sin\alpha_2}\leq\frac{d_1+d_2}{d_1+d_2+d_3}.
\]

\tl
\end{Thm}
\begin{figure}
\begin{center}
\input{cl_typeII_point.tex}\\[1.5em]
(a) $L\cap S$ is a single point\\[1.5em]
\input{cl_typeII_edge.tex}\\[1.5em]
(b) $L\cap S$ is an edge\\[1.5em]
\end{center}
\caption{Figure for Theorem~\ref{TypeII_CL}}
\label{CL_TypeII_figure}
\end{figure}

\begin{proof}
First, if $v_3$ has degree 1 in $F$ then we can rotate $L$ about $v_1$ toward $v_3$'s component, decreasing the length of the edge adjacent to $v_3$ while adding an edge from $v_2$ perpendicular to $L$. This results in an OC of smaller length, giving a contradiction.

For $(a)$ and $(b)$, if we rotate $L$ about $v_1$ either way, moving $v_2$ and $v_3$ and adjusting $F$ accordingly, it follows from minimality that the length of $F$ does not decrease.  The rate of change for this rotation is $\pm(d_2 - 2(d_2+d_3)\sin\alpha_3)$  for $(a)$ and $\pm(2d_2\sin\alpha_2 - 2(d_2+d_3)\sin\alpha_3)$ for $(b)$.  The two equalities follow.

For $(c)$, if we rotate $L$ about $v_1'$ towards $v_2$'s component, moving $v_2$ and $v_3$ and adjusting $F$ accordingly, again the length of $F$ does not decrease.  The rate of change for this rotation is $2(d_2+d_3)\sin\alpha_3-d_2$, which must be nonnegative.
Similarly , if we rotate $L$ about $v_1$ towards $v_3$'s component, the rate of change is $(d_1+d_2)-2(d_1+d_2+d_3)\sin\alpha_3$, which must also be nonnegative.  The inequalities follow.

For $(d)$, if we rotate $L$ the same way as for $(c)$ we get rates of  change of $2(d_2+d_3)\sin\alpha_3 - 2d_2\sin\alpha_2$ and  $2(d_1+d_2)\sin\alpha_2 - 2(d_1+d_2+d_3)\sin\alpha_3$ respectively, which must both also be nonnegative.  The inequalities follow.
\end{proof}

\begin{Thm}\label{TypeIII_lemma}  Let $L$ be a Type III critical line for which all incident points of $\cal F$ are FCPs.  Then referring to Figure~\ref{TypeIII_figure}, both $v_1$ and $v_3$ have degree 2 in $L$ and:
\begin{enumerate}\renewcommand{\labelenumi}{$(\alph{enumi})$}
\item
if $v_2$ has degree 1 in $F$, then
\[
\sin\alpha_1=\frac{1}{2}\left(\frac{d_3}{d_1+d_3}\right),\ \
\sin\alpha_3=\frac{1}{2}\left(\frac{d_1}{d_1+d_3}\right);
\]
 \item
if $v_2$ has degree 2 in $F$, then
\[
\sin\alpha_1=\frac{d_3}{d_1+d_3}\sin\alpha_2,\ \
\sin\alpha_3=\frac{d_1}{d_1+d_3}\sin\alpha_2.
\]
\end{enumerate}
\end{Thm}
\begin{figure}
\begin{center}
\input{cl_typeIII.tex}\\[1.5em]
\end{center}
\caption{Figure for Theorem~\ref{TypeIII_lemma}}
\label{TypeIII_figure}
\end{figure}
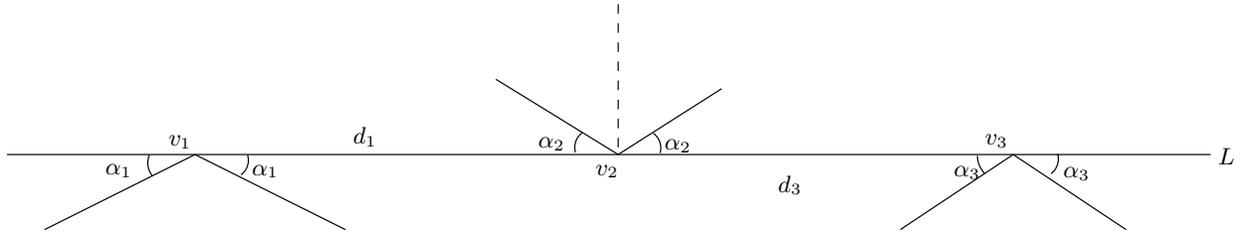

\begin{proof} That $v_1$ and $v_3$ both have degree 2 follows from the same argument as given in the proof of Theorem~\ref{TypeII_CL}.
\\[1em]
For $(a)$, consider the perturbation of $F$ consisting of rotating $L$ clockwise about a point on $L$ distance $x$  to the right from $v_1$, moving $v_1$, $v_2$, and $v_3$ along the associated arcs, and moving the endpoints of the adjacent edges correspondingly.  This causes a total rate of change of length for $F$ of
\[
2 x\sin \alpha_1 +(d_1-x) -2(d_1+d_3-x)\sin \alpha_3.
\]
Since $F$ has minimum length,  this perturbation cannot decrease the length of $F$, and so we must have this value equal to 0 for {\em every}\/ value of $x$.  Evaluating the perturbation at $x=d_1$ and $x=0$, we obtain
\[
d_1\sin \alpha_1-d_3\sin \alpha_3=0,\ \
d_1 -2(d_1+d_3)\sin\alpha_3=0
\]
from which the equations follow.
\\[1em]
For $(b)$, if $v_2$ has degree 2, we rotate $L$ the same way, now moving the edges adjacent to $v_2$ accordingly.  The rate of change of length is now
\[
2x\sin \alpha_1 +2(d_1-x)\sin\alpha_2 -2(d_1+d_3-x)\sin \alpha_3.
\]
Again, since this must have value 0 for every $x$ we have
\[
d_1\sin \alpha_1-d_3\sin \alpha_3=0,
\ \
d_1\sin\alpha_2 -(d_1+d_3)\sin\alpha_3=0
\]
from which the equations follow.
\end{proof}

\medskip

We end this section by considering the OCP for the set $S$ with boundary polygon $ABCD$ given in Figure~\ref{2C_OC_Figure}.
\begin{figure}
   \begin{center}
       \epsfxsize=5in
       \epsffile{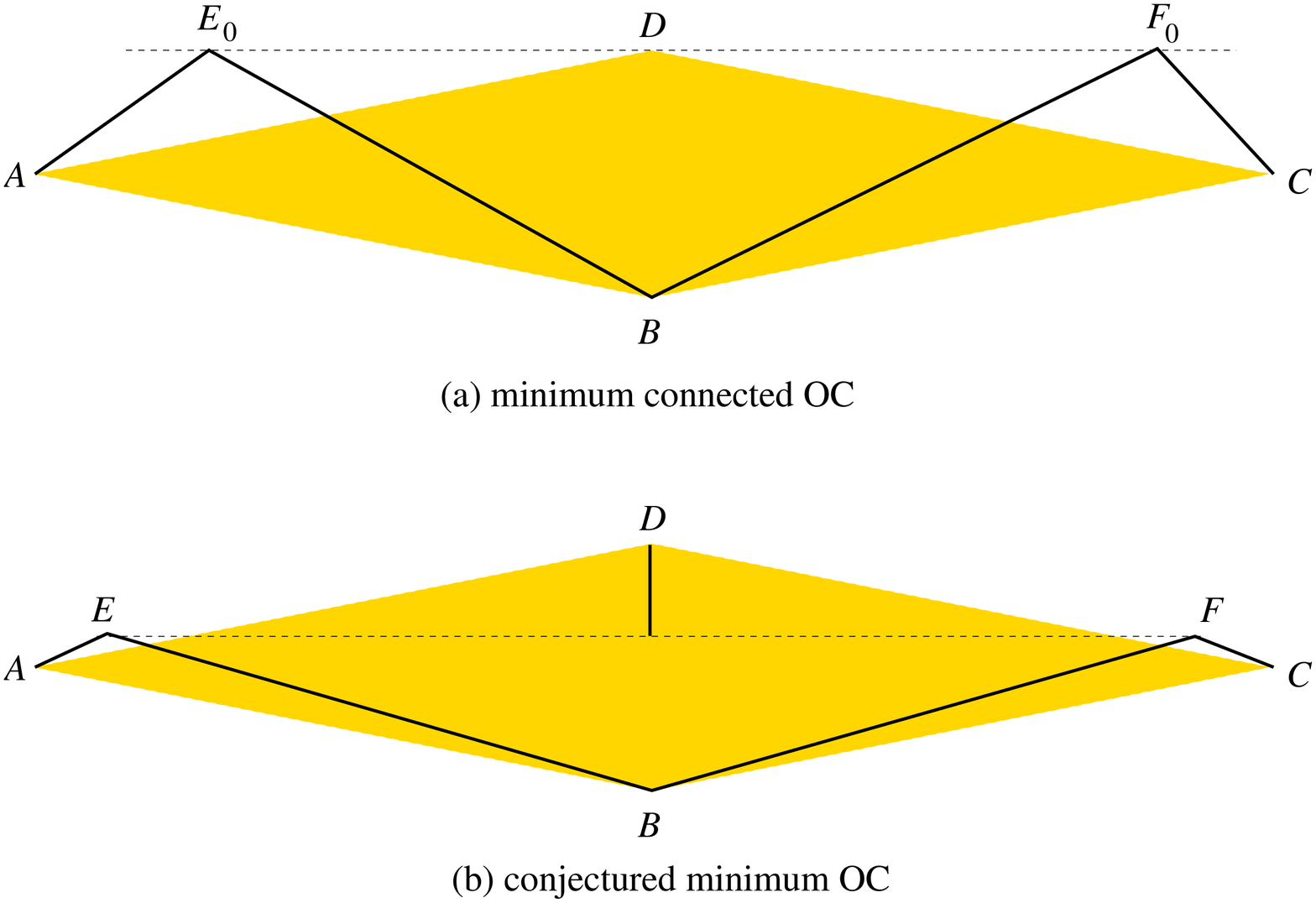}
   \end{center}
\vspace*{0in}
\centerline{Coordinate Values}\ \\[-2em]
\[\begin{array}{|c|rr|}
\hline
A &     0.0000 &     1.0000\\
B &     8.0000 &     0.0000\\
C &    16.0000 &     1.0000\\
D &     8.0000 &     2.0000\\
E &     2.0635 &     1.5328\\
F &    13.9365 &     1.5328\\
E_0 &     2.6667 &     2.0000\\
F_0 &    13.3333 &     2.0000\\
E' &     2.4259 &     1.3032\\
F' &    13.5741 &     1.3032\\
\hline
\end{array}\]
\caption{Two opaque covers}\label{2C_OC_Figure}
\end{figure}
Below are the relevant OCP solutions and conjectured solutions for $S$:
\begin{itemize}
\item
The minimum interior single-path OC is the path $ABDC$, with length 18.1246.
\item
By Corollary~\ref{connected interior OC}, the minimum interior connected OC is the MST on $ABCD$; this is a full MST with length 17.7602.
\item
The conjectured minimum connected OC, which happens to be a single-path OC, is given in Figure~\ref{2C_OC_Figure}(a), and consists of the path $AE_0BF_0C$ with $E_0=v_1$ and $F_0=v_3$ satisfying the conditions of Theorem~\ref{boundary-critical_connected}(d). The length of this path is 17.1245.
\item
The conjectured minimum OC for $S$ is the one shown in Figure~\ref{2C_OC_Figure}(b).  It consists of two components, one of which is the path $AEBFC$ going through the points $E$ and $F$ determined by Theorem~\ref{TypeIII_lemma}(a).  Its length is 16.9919.
\item
The above conjectured minimum OC lies partly outside $S$.  We can compute a conjectured minimum interior OC for $S$ by determining the points $E'$ on $AD$ and $F'$ on $CD$ that minimize the total length of the path $AE'BF'C$ plus the length of the perpendicular from $D$ to $E'F'$. The optimal locations for $E'$ and $F'$ are given in the table in Figure~\ref{2C_OC_Figure}.  The  length of this OC is 17.0352, which is  less than any of the connected OCs.
\item
This example also shows that the construction of Shermer \cite{Shermer}, discussed at the beginning of Section~5, does not always give a solution to the OCP.  Here the Shermer proposed OC turns out to consist of the MST on $ABC$, which is simply the path $ABC$, together with the perpendicular from $D$ to $AC$.  This has length 17.1245, which is greater than either of the 2-component OCs given above.
\item
The line through the points $E$ and $F$ in Figure~\ref{2C_OC_Figure}(b) is a Type III critical line.  We note that if we replace $D$ by the point $D'$ with the same $x$-coordinate as $D$ and with $y$-coordinate  changed to 1.5328 (that of $E$ and $F$), then the conjectured minimum OC for the resulting polygon $ABCD'$ is now the single path $AEBFC$, and the line through $E$ and $F$ is now a Type I critical line.  We have not been able to find a conjectured minimum OC with a Type II critical line that is shorter than ones without such critical lines, and it is an open question whether such critical lines actually occur in minimum OCs.
\end{itemize}

\subsection{Further properties of OCs with multiple connected components}

For the general OCP, we give further structural properties for each connected component, relating to its structure as an MST.   Let $F_0$ be a connected component of an OC, let $e$ be an edge of $F_0$, and let $p$ be an
interior point of $e$. We say that $p$ is \emph{visible} if for
each of the four quadrants determined by the line extending $e$
and the line perpendicular to $e$ through $p$ there exists a ray
from $p$ lying in that quadrant not intersecting $F_0$ (at a point
other than $p$). If there is no such ray in one or more of the
quadrants, then $p$ is said to be \emph{invisible}.  An edge of $F_0$ is said to be \emph{fully visible} if there is
no subinterval of $e$ (of length $>0$) consisting only of
invisible points (with respect to the same quadrant).  For example, in Figure~\ref{splitedge}(a) the edge $e=t_1s_1$ is not fully visible
since the interval $ps_1$ consists only of invisible points.
\begin{figure}[ht]
\begin{center}\includegraphics[width=5in]{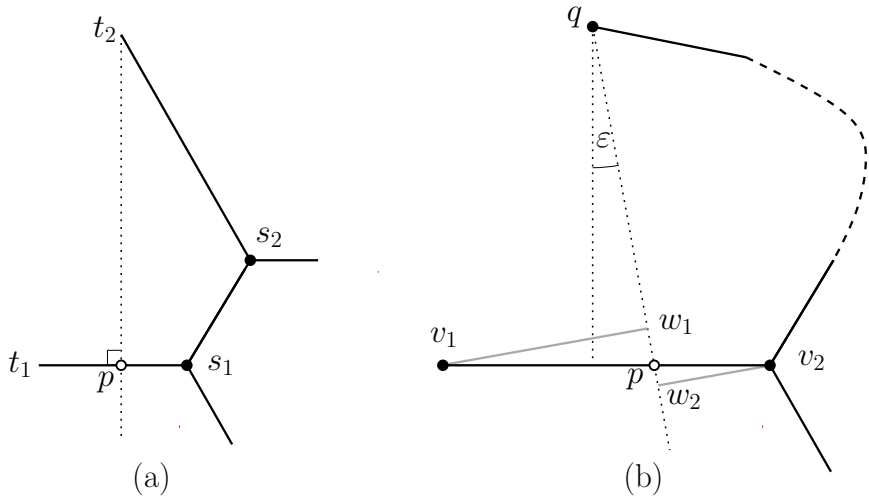}
\end{center}
\caption{Examples of invisible points, $p$, belonging, in each case to an edge of an OC that is not fully visible. In (b) the edge  $e=v_1v_2$ is not minimal because $e$ can be replaced by a  pair of edges of shorter length.}\label{splitedge}
\end{figure}

\begin{Thm}
\label{thm:fullyvisible} If $F$ is a minimum graphical OC, then
every edge of $F$ is fully visible.
\end{Thm}

\begin{proof} Suppose the edge $e=v_1v_2$  of a graphical OC $F$ is not fully visible. We will
show that $F$ is non-minimal.

Let $F_0$ be the connected component of $F$ containing $e$. Without loss of generality, we can assume that $e$ is
horizontal. Choose a subinterval $I$ of $e$ such that every
point of $I$ is invisible with respect to, say, the top right
quadrant.

Let $p$ be an interior point of $I$. Then there exists an angle
$\varepsilon >0$ such that the ray from $p$ at an angle of
$\varepsilon$ to the left of the upwards vertical ray from $p$
meets $F$ at a point $q$, and the projection of $pq$ onto $e$
lies in $I$. The existence of such an $\varepsilon$ and $q$
follow from the fact that all points in the interval $I$ to the
left of $p$ are also invisible (see Figure~\ref{splitedge}(b)).
The OC
$F$ can be shortened as follows. Replace $v_1v_2$ by two line
segments $v_1w_1$ and $w_2v_2$ where each $w_i$ is the
projection of $v_i$ onto the line extending $pq$. Clearly
$|v_1w_1|+|w_2v_2|<|v_1v_2|$, since $\varepsilon >0$. But $F\cup
v_1w_1 \cup w_2v_2\setminus v_1v_2$ is also an OC
since every line that intersects $w_1w_2$ meets $F$ by
construction.
 \end{proof}

\begin{Cor}
\label{cor:longsteiner} Let $F$ be a minimum OC. Let
$T$ be a full Steiner component of
 $F$ containing two adjacent Steiner vertices $s_1$ and $s_2$. Then no edge of $T$ incident with $s_1$
or $s_2$ is longer than $s_1s_2$.
\end{Cor}

\begin{proof} Suppose there exists a vertex $t_2$ adjacent to $s_2$ such that $|t_2s_2| > |s_1s_2|$.
Let $t_1$ be a vertex adjacent to $s_1$ on the same side of the line through $s_1s_2$ as $t_2$.
Since each of the angles between incident edges at $s_1$ and $s_2$ is $2\pi/3$ it follows that
$t_1s_1$ is not fully visible (as in Figure~\ref{splitedge}(a)), contradicting the minimality of
$F$. \end{proof}

\begin{Cor}
\label{cor:2steinerpts} Let $F$ be a minimum OC. If
$T$ is a full Steiner component of
 $F$ then $T$ contains at most two Steiner vertices.
\end{Cor}

\begin{proof} Suppose $T$ contains three Steiner vertices $s_1$, $s_2$,
$s_3$, such that $s_1$ and $s_3$ are both adjacent to $s_2$. By
Corollary~\ref{cor:longsteiner}, $|s_1s_2|=|s_2s_3|$. But then
one of the other edges incident with $s_1$ is not fully visible.
\end{proof}

The following corollary is also an immediate consequence of Theorem~\ref{thm:fullyvisible} and is a potentially useful tool for
establishing the non-minimality of certain OCs. For example, in the OC in Figure~\ref{connectedOC} the critical lines through $AL$ and $IK$ both fail
the condition of the corollary, implying that the OC is not a minimum graphical OC for the associated region $S$. This is verified by the shorter
OC illustrated in Figure~\ref{disconnectedOC}.

\begin{Cor}
\label{cor:1critical} If an edge of a component of  a minimum OC
meets a critical line $L$ at an angle of $\pi/2$, then no other edge of that component meets $L$.
\end{Cor}

We finish the section by revisiting the OCP for a triangle.  Theorem~\ref{contriangle} establishes the structure of the minimum connected OC.  For the graphical OCP it is not even known if there is a bound on the number of components a solution can have, let alone the structure of such a solution.   We end the paper by presenting one result regarding minimum graphical OCPs on a triangle.

\begin{Thm}\label{bd2triangle}
For any  triangular region $S=\triangle abc$, the minimum OC on $S$ having at most two components is the MST on $\{a,b,c\}$.
\end{Thm}

\begin{proof}
Let $F$ be a minimum 2-component OC for $S=\triangle abc$, and, without loss of generality, suppose that component $C_a$
contains $a$ and component $C_{bc}$ contains $b$ and $c$.
Let $L$ be the critical line separating the two connected components, with consecutive tangent points $v_1, v_2, v_3$, and note that this must be the only critical line, and it will be of Type III. By  Lemma~\ref{CL_lemma} $v_1$ and $v_3$ both belong to one connected component of $F$, and $v_2$ to the other.

Suppose $v_1$ and $v_3$ both belong to $C_{bc}$. Then  $C_{bc}$ is an OC for $\triangle v_2 bc$ (since $v_2$ lies on the line segment $v_1v_3$). If $T_2$ is a MST for $\triangle v_2 bc$ then, by Theoerem~\ref{contriangle}, $|T_2| \leq |C_{bc}|$, and $T_2 \cup C_a$ is a connected OC for $S$, contradicting the minimality of $F$.

Similarly, if $v_1$ and $v_3$ both belong to $C_{a}$, then $|v_2a|\leq |C_{a}|$, and again $F$ is non-minimal.

Thus $F$ can have only one component, and so by Theorem~\ref{contriangle} must be a MST on $S$.
\end{proof}

It is possible to prove stronger versions of Theorem~\ref{bd2triangle}, for OCs with more than two components (for example, the case where there are up to three components). However the proofs appear to become increasingly technical, and involve eliminating possibilities where some critical points are not free (such as the OC in Figure~\ref{nonfree}(a)). Currently there is no known method for generalising Theorem~\ref{bd2triangle} for an arbitrarily large number of components. Nevertheless, we believe the following conjecture to be true.
\begin{Conj}
The minimum graphical OC for a triangle is the MST on its vertices.
\end{Conj}

\section{Conclusion}

In this paper we have studied the Opaque Cover Problem (OCP) for any polygonal region $S$ in the Euclidean plane. This is the problem of finding a minimum-length set $F$ which intersects every straight line passing through $S$, where length is defined using the 1-dimensional Hausdorff measure. A key conjecture relating to this problem is the following:

\noindent  \textbf{The Graphical Conjecture:} A minimum OC for a polygonal region $S$ in the Euclidean plane is graphical, that is, composed of a finite number of straight-line segments.

In light of this conjecture, we have limited ourselves to the study of minimum graphical OCs. The main part of this study is divided into two sections, first a study of minimum connected graphical OCs (Section~4) and then minimum graphical OCs with multiple components (Section~5). In each of these sections we characterise  structural properties of $F$ for each component, and strong geometric constraints on the locations of the extreme points of each component.

The most important structural property is that each connected component is a Steiner minimum tree on its extreme vertices, and there is a limit on the number of Steiner vertices in each component. Not only does this constrain the topology of each component of $F$, but it also allows us to make use of efficient software packages, such as GeoSteiner, for solving the Steiner tree problem to explicitly construct the components of $F$ once their extreme vertices are known.

The extreme points of each component have been shown to be either vertices of $\partial S$  or critical points. The locations of the critical points and the directions of their incident edges in $F$ are strongly constrained by  properties of the critical lines of $F$. Some of the key results of this type are those in Theorems~\ref{boundary-critical_connected}, \ref{TypeI_CL}, \ref{TypeII_CL} and \ref{TypeIII_lemma}. These theorems can be thought of as providing a mathematical toolbox for attempting to construct minimum OCs. The conjectured minimum connected solution for the example in Figure~\ref{connectedOC}, and both conjectured solutions in Figure~\ref{2C_OC_Figure} were constructed with the aid of the properties in these theorems.

Despite the properties established in this paper, our understanding of minimum OCs remains frustratingly incomplete. One of our long-term aims is to design an exact finite algorithm for constructing a minimum OC for any polygonal region. In the shorter term there are a number of key open questions, including The Graphical Conjecture above, still to be addressed:
\begin{enumerate}
  \item \emph{Is the minimum graphical OC for a triangular region $S$ the MST on $V_S$?} We conjecture that this is the case, and have proved it for OCs containing at most two connected components. It is currently unclear how to generalise this result to an arbitrarily large number of connected components. Answering this question may depend on being able to answer the next open question below.
  \item \emph{In a minimum graphical OC, under what conditions can there exist critical points that are not free?} All of the conjectured minimum OCs that the authors have seen have the property that any critical point that is not free is determined by two Type~I critical lines that are extensions of edges of $S$. It would be useful to resolve whether or not this always holds, since the geometric constraints  in this paper on the locations of the critical points of a minimum OC only apply to FCPs.
  \item \emph{Do critical lines of Type~II exist for any minimum graphical OC?} There are no known examples of conjectured minimum OCs containing a Type~II critical line.
\end{enumerate}

\end{document}

%% file: cl_boundary.tex
\begin{picture}(0,0)%
\includegraphics{cl_boundary.pstex}%
\end{picture}%
\setlength{\unitlength}{3947sp}%
\begingroup\makeatletter\ifx\SetFigFont\undefined%
\gdef\SetFigFont#1#2#3#4#5{%
  \reset@font\fontsize{#1}{#2pt}%
  \fontfamily{#3}\fontseries{#4}\fontshape{#5}%
  \selectfont}%
\fi\endgroup%
\begin{picture}(5522,1544)(1083,-2378)
\put(3692,-2063){\makebox(0,0)[lb]{\smash{{\SetFigFont{9}{10.8}{\rmdefault}{\mddefault}{\updefault}{\color[rgb]{0,0,0}$S$}%
}}}}
\put(1264,-1048){\makebox(0,0)[lb]{\smash{{\SetFigFont{9}{10.8}{\rmdefault}{\mddefault}{\updefault}{\color[rgb]{0,0,0}$L$}%
}}}}
\put(3692,-2063){\makebox(0,0)[lb]{\smash{{\SetFigFont{9}{10.8}{\rmdefault}{\mddefault}{\updefault}{\color[rgb]{0,0,0}$S$}%
}}}}
\put(2739,-922){\makebox(0,0)[lb]{\smash{{\SetFigFont{9}{10.8}{\rmdefault}{\mddefault}{\updefault}{\color[rgb]{0,0,0}$d_1$}%
}}}}
\put(2463,-1096){\makebox(0,0)[lb]{\smash{{\SetFigFont{9}{10.8}{\rmdefault}{\mddefault}{\updefault}{\color[rgb]{0,0,0}$\alpha_1$}%
}}}}
\put(1750,-1109){\makebox(0,0)[lb]{\smash{{\SetFigFont{9}{10.8}{\rmdefault}{\mddefault}{\updefault}{\color[rgb]{0,0,0}$\alpha_1$}%
}}}}
\put(5715,-1109){\makebox(0,0)[lb]{\smash{{\SetFigFont{9}{10.8}{\rmdefault}{\mddefault}{\updefault}{\color[rgb]{0,0,0}$\alpha_3$}%
}}}}
\put(2089,-918){\makebox(0,0)[lb]{\smash{{\SetFigFont{9}{10.8}{\rmdefault}{\mddefault}{\updefault}{\color[rgb]{0,0,0}$v_1$}%
}}}}
\put(3737,-918){\makebox(0,0)[lb]{\smash{{\SetFigFont{9}{10.8}{\rmdefault}{\mddefault}{\updefault}{\color[rgb]{0,0,0}$v_2$}%
}}}}
\put(5372,-918){\makebox(0,0)[lb]{\smash{{\SetFigFont{9}{10.8}{\rmdefault}{\mddefault}{\updefault}{\color[rgb]{0,0,0}$v_3$}%
}}}}
\put(4732,-934){\makebox(0,0)[lb]{\smash{{\SetFigFont{9}{10.8}{\rmdefault}{\mddefault}{\updefault}{\color[rgb]{0,0,0}$d_3$}%
}}}}
\put(5109,-1096){\makebox(0,0)[lb]{\smash{{\SetFigFont{9}{10.8}{\rmdefault}{\mddefault}{\updefault}{\color[rgb]{0,0,0}$\alpha_3$}%
}}}}
\end{picture}%

%% file: cl_boundary_edge.tex
\begin{picture}(0,0)%
\includegraphics{cl_boundary_edge.pstex}%
\end{picture}%
\setlength{\unitlength}{3947sp}%
\begingroup\makeatletter\ifx\SetFigFont\undefined%
\gdef\SetFigFont#1#2#3#4#5{%
  \reset@font\fontsize{#1}{#2pt}%
  \fontfamily{#3}\fontseries{#4}\fontshape{#5}%
  \selectfont}%
\fi\endgroup%
\begin{picture}(5301,1505)(1355,-2294)
\put(3688,-1995){\makebox(0,0)[lb]{\smash{{\SetFigFont{9}{10.8}{\rmdefault}{\mddefault}{\updefault}{\color[rgb]{0,0,0}$S$}%
}}}}
\put(3688,-1995){\makebox(0,0)[lb]{\smash{{\SetFigFont{9}{10.8}{\rmdefault}{\mddefault}{\updefault}{\color[rgb]{0,0,0}$S$}%
}}}}
\put(1358,-1020){\makebox(0,0)[lb]{\smash{{\SetFigFont{9}{10.8}{\rmdefault}{\mddefault}{\updefault}{\color[rgb]{0,0,0}$L$}%
}}}}
\put(3688,-1995){\makebox(0,0)[lb]{\smash{{\SetFigFont{9}{10.8}{\rmdefault}{\mddefault}{\updefault}{\color[rgb]{0,0,0}$S$}%
}}}}
\put(2774,-899){\makebox(0,0)[lb]{\smash{{\SetFigFont{9}{10.8}{\rmdefault}{\mddefault}{\updefault}{\color[rgb]{0,0,0}$d_1$}%
}}}}
\put(2508,-1067){\makebox(0,0)[lb]{\smash{{\SetFigFont{9}{10.8}{\rmdefault}{\mddefault}{\updefault}{\color[rgb]{0,0,0}$\alpha_1$}%
}}}}
\put(1824,-1079){\makebox(0,0)[lb]{\smash{{\SetFigFont{9}{10.8}{\rmdefault}{\mddefault}{\updefault}{\color[rgb]{0,0,0}$\alpha_1$}%
}}}}
\put(5631,-1079){\makebox(0,0)[lb]{\smash{{\SetFigFont{9}{10.8}{\rmdefault}{\mddefault}{\updefault}{\color[rgb]{0,0,0}$\alpha_3$}%
}}}}
\put(2149,-895){\makebox(0,0)[lb]{\smash{{\SetFigFont{9}{10.8}{\rmdefault}{\mddefault}{\updefault}{\color[rgb]{0,0,0}$v_1$}%
}}}}
\put(5301,-895){\makebox(0,0)[lb]{\smash{{\SetFigFont{9}{10.8}{\rmdefault}{\mddefault}{\updefault}{\color[rgb]{0,0,0}$v_3$}%
}}}}
\put(4687,-911){\makebox(0,0)[lb]{\smash{{\SetFigFont{9}{10.8}{\rmdefault}{\mddefault}{\updefault}{\color[rgb]{0,0,0}$d_3$}%
}}}}
\put(5048,-1067){\makebox(0,0)[lb]{\smash{{\SetFigFont{9}{10.8}{\rmdefault}{\mddefault}{\updefault}{\color[rgb]{0,0,0}$\alpha_3$}%
}}}}
\put(3479,-907){\makebox(0,0)[lb]{\smash{{\SetFigFont{9}{10.8}{\rmdefault}{\mddefault}{\updefault}{\color[rgb]{0,0,0}$v_2$}%
}}}}
\put(4235,-907){\makebox(0,0)[lb]{\smash{{\SetFigFont{9}{10.8}{\rmdefault}{\mddefault}{\updefault}{\color[rgb]{0,0,0}$v_2'$}%
}}}}
\put(3883,-918){\makebox(0,0)[lb]{\smash{{\SetFigFont{9}{10.8}{\rmdefault}{\mddefault}{\updefault}{\color[rgb]{0,0,0}$d_2$}%
}}}}
\end{picture}%

%% file: BD_component_figure.tex
\begin{picture}(0,0)%
\includegraphics{BD_component_figure.pstex}%
\end{picture}%
\setlength{\unitlength}{3947sp}%
\begingroup\makeatletter\ifx\SetFigFont\undefined%
\gdef\SetFigFont#1#2#3#4#5{%
  \reset@font\fontsize{#1}{#2pt}%
  \fontfamily{#3}\fontseries{#4}\fontshape{#5}%
  \selectfont}%
\fi\endgroup%
\begin{picture}(4262,3747)(1710,-408)
\put(2776,2505){\makebox(0,0)[lb]{\smash{{\SetFigFont{9}{10.8}{\rmdefault}{\mddefault}{\updefault}{\color[rgb]{0,0,0}$S$}%
}}}}
\put(3685,-96){\makebox(0,0)[lb]{\smash{{\SetFigFont{9}{10.8}{\rmdefault}{\mddefault}{\updefault}{\color[rgb]{0,0,0}$S$}%
}}}}
\put(5773,2949){\makebox(0,0)[lb]{\smash{{\SetFigFont{9}{10.8}{\rmdefault}{\mddefault}{\updefault}{\color[rgb]{0,0,0}$L$}%
}}}}
\put(5610,670){\makebox(0,0)[lb]{\smash{{\SetFigFont{9}{10.8}{\rmdefault}{\mddefault}{\updefault}{\color[rgb]{0,0,0}$L$}%
}}}}
\put(4241,3040){\makebox(0,0)[lb]{\smash{{\SetFigFont{9}{10.8}{\rmdefault}{\mddefault}{\updefault}{\color[rgb]{0,0,0}$u$}%
}}}}
\put(2753,3070){\makebox(0,0)[lb]{\smash{{\SetFigFont{9}{10.8}{\rmdefault}{\mddefault}{\updefault}{\color[rgb]{0,0,0}$v$}%
}}}}
\put(3999,803){\makebox(0,0)[lb]{\smash{{\SetFigFont{9}{10.8}{\rmdefault}{\mddefault}{\updefault}{\color[rgb]{0,0,0}$u$}%
}}}}
\put(3064,554){\makebox(0,0)[lb]{\smash{{\SetFigFont{9}{10.8}{\rmdefault}{\mddefault}{\updefault}{\color[rgb]{0,0,0}$v$}%
}}}}
\put(4204,2478){\makebox(0,0)[lb]{\smash{{\SetFigFont{9}{10.8}{\rmdefault}{\mddefault}{\updefault}{\color[rgb]{0,0,0}$\bar C_1$}%
}}}}
\put(3963,259){\makebox(0,0)[lb]{\smash{{\SetFigFont{9}{10.8}{\rmdefault}{\mddefault}{\updefault}{\color[rgb]{0,0,0}$\bar C_1$}%
}}}}
\put(2972,1084){\makebox(0,0)[lb]{\smash{{\SetFigFont{9}{10.8}{\rmdefault}{\mddefault}{\updefault}{\color[rgb]{0,0,0}$\bar C_2$}%
}}}}
\end{picture}%

%% file: CL_lemma_figure.tex
\begin{picture}(0,0)%
\includegraphics{CL_lemma_figure.pstex}%
\end{picture}%
\setlength{\unitlength}{3947sp}%
\begingroup\makeatletter\ifx\SetFigFont\undefined%
\gdef\SetFigFont#1#2#3#4#5{%
  \reset@font\fontsize{#1}{#2pt}%
  \fontfamily{#3}\fontseries{#4}\fontshape{#5}%
  \selectfont}%
\fi\endgroup%
\begin{picture}(8194,5271)(299,-2171)
\put(5564,1204){\makebox(0,0)[lb]{\smash{{\SetFigFont{10}{12.0}{\rmdefault}{\mddefault}{\updefault}{\color[rgb]{0,0,0}$v_1$}%
}}}}
\put(6392,1204){\makebox(0,0)[lb]{\smash{{\SetFigFont{10}{12.0}{\rmdefault}{\mddefault}{\updefault}{\color[rgb]{0,0,0}$v_2$}%
}}}}
\put(7016,964){\makebox(0,0)[lb]{\smash{{\SetFigFont{10}{12.0}{\rmdefault}{\mddefault}{\updefault}{\color[rgb]{0,0,0}$v_3$}%
}}}}
\put(8247,1093){\makebox(0,0)[lb]{\smash{{\SetFigFont{10}{12.0}{\rmdefault}{\mddefault}{\updefault}{\color[rgb]{0,0,0}$L$}%
}}}}
\put(5612,757){\makebox(0,0)[lb]{\smash{{\SetFigFont{10}{12.0}{\rmdefault}{\mddefault}{\updefault}{\color[rgb]{0,0,0}$S$}%
}}}}
\put(2529,960){\makebox(0,0)[lb]{\smash{{\SetFigFont{10}{12.0}{\rmdefault}{\mddefault}{\updefault}{\color[rgb]{0,0,0}$v_3$}%
}}}}
\put(3979,1105){\makebox(0,0)[lb]{\smash{{\SetFigFont{10}{12.0}{\rmdefault}{\mddefault}{\updefault}{\color[rgb]{0,0,0}$L$}%
}}}}
\put(1977,1200){\makebox(0,0)[lb]{\smash{{\SetFigFont{10}{12.0}{\rmdefault}{\mddefault}{\updefault}{\color[rgb]{0,0,0}$v_2$}%
}}}}
\put(1211,733){\makebox(0,0)[lb]{\smash{{\SetFigFont{10}{12.0}{\rmdefault}{\mddefault}{\updefault}{\color[rgb]{0,0,0}$S$}%
}}}}
\put(1173,1212){\makebox(0,0)[lb]{\smash{{\SetFigFont{10}{12.0}{\rmdefault}{\mddefault}{\updefault}{\color[rgb]{0,0,0}$v_1$}%
}}}}
\put(3199,-2117){\makebox(0,0)[lb]{\smash{{\SetFigFont{12}{14.4}{\rmdefault}{\mddefault}{\updefault}{\color[rgb]{0,0,0}Produces Type III critical lines}%
}}}}
\put(3214,178){\makebox(0,0)[lb]{\smash{{\SetFigFont{12}{14.4}{\rmdefault}{\mddefault}{\updefault}{\color[rgb]{0,0,0}Produces Type II critical lines}%
}}}}
\put(4145,2959){\makebox(0,0)[lb]{\smash{{\SetFigFont{10}{12.0}{\rmdefault}{\mddefault}{\updefault}{\color[rgb]{0,0,0}$v_2$}%
}}}}
\put(5096,2935){\makebox(0,0)[lb]{\smash{{\SetFigFont{10}{12.0}{\rmdefault}{\mddefault}{\updefault}{\color[rgb]{0,0,0}$v_3$}%
}}}}
\put(3264,2968){\makebox(0,0)[lb]{\smash{{\SetFigFont{10}{12.0}{\rmdefault}{\mddefault}{\updefault}{\color[rgb]{0,0,0}$v_1$}%
}}}}
\put(6597,2862){\makebox(0,0)[lb]{\smash{{\SetFigFont{10}{12.0}{\rmdefault}{\mddefault}{\updefault}{\color[rgb]{0,0,0}$L$}%
}}}}
\put(3214,1963){\makebox(0,0)[lb]{\smash{{\SetFigFont{12}{14.4}{\rmdefault}{\mddefault}{\updefault}{\color[rgb]{0,0,0}Produces Type I critical lines}%
}}}}
\put(2284,-786){\makebox(0,0)[lb]{\smash{{\SetFigFont{10}{12.0}{\rmdefault}{\mddefault}{\updefault}{\color[rgb]{0,0,0}$v_2$}%
}}}}
\put(2783,-1028){\makebox(0,0)[lb]{\smash{{\SetFigFont{10}{12.0}{\rmdefault}{\mddefault}{\updefault}{\color[rgb]{0,0,0}$v_3$}%
}}}}
\put(4072,-871){\makebox(0,0)[lb]{\smash{{\SetFigFont{10}{12.0}{\rmdefault}{\mddefault}{\updefault}{\color[rgb]{0,0,0}$L$}%
}}}}
\put(5567,-992){\makebox(0,0)[lb]{\smash{{\SetFigFont{10}{12.0}{\rmdefault}{\mddefault}{\updefault}{\color[rgb]{0,0,0}$v_1$}%
}}}}
\put(5968,-799){\makebox(0,0)[lb]{\smash{{\SetFigFont{10}{12.0}{\rmdefault}{\mddefault}{\updefault}{\color[rgb]{0,0,0}$v_2$}%
}}}}
\put(1643,-1016){\makebox(0,0)[lb]{\smash{{\SetFigFont{10}{12.0}{\rmdefault}{\mddefault}{\updefault}{\color[rgb]{0,0,0}$v_1$}%
}}}}
\put(7997,-871){\makebox(0,0)[lb]{\smash{{\SetFigFont{10}{12.0}{\rmdefault}{\mddefault}{\updefault}{\color[rgb]{0,0,0}$L$}%
}}}}
\put(4175,2470){\makebox(0,0)[lb]{\smash{{\SetFigFont{10}{12.0}{\rmdefault}{\mddefault}{\updefault}{\color[rgb]{0,0,0}$S$}%
}}}}
\put(2304,-1507){\makebox(0,0)[lb]{\smash{{\SetFigFont{10}{12.0}{\rmdefault}{\mddefault}{\updefault}{\color[rgb]{0,0,0}$S$}%
}}}}
\put(6205,-1483){\makebox(0,0)[lb]{\smash{{\SetFigFont{10}{12.0}{\rmdefault}{\mddefault}{\updefault}{\color[rgb]{0,0,0}$S$}%
}}}}
\put(6312,-1004){\makebox(0,0)[lb]{\smash{{\SetFigFont{10}{12.0}{\rmdefault}{\mddefault}{\updefault}{\color[rgb]{0,0,0}$v_3$}%
}}}}
\end{picture}%

%% file: CL_freepoint_figure.tex
\begin{picture}(0,0)%
\includegraphics{CL_freepoint_figure.pstex}%
\end{picture}%
\setlength{\unitlength}{3947sp}%
\begingroup\makeatletter\ifx\SetFigFont\undefined%
\gdef\SetFigFont#1#2#3#4#5{%
  \reset@font\fontsize{#1}{#2pt}%
  \fontfamily{#3}\fontseries{#4}\fontshape{#5}%
  \selectfont}%
\fi\endgroup%
\begin{picture}(6011,3336)(-271,-2486)
\put(4317,630){\makebox(0,0)[lb]{\smash{{\SetFigFont{10}{12.0}{\rmdefault}{\mddefault}{\updefault}{\color[rgb]{0,0,0}$v_2$}%
}}}}
\put(4869,390){\makebox(0,0)[lb]{\smash{{\SetFigFont{10}{12.0}{\rmdefault}{\mddefault}{\updefault}{\color[rgb]{0,0,0}$v_3$}%
}}}}
\put(3513,642){\makebox(0,0)[lb]{\smash{{\SetFigFont{10}{12.0}{\rmdefault}{\mddefault}{\updefault}{\color[rgb]{0,0,0}$v_1$}%
}}}}
\put(3551,163){\makebox(0,0)[lb]{\smash{{\SetFigFont{10}{12.0}{\rmdefault}{\mddefault}{\updefault}{\color[rgb]{0,0,0}$S$}%
}}}}
\put(5494,490){\makebox(0,0)[lb]{\smash{{\SetFigFont{10}{12.0}{\rmdefault}{\mddefault}{\updefault}{\color[rgb]{0,0,0}$L$}%
}}}}
\put(3574,-347){\makebox(0,0)[lb]{\smash{{\SetFigFont{12}{14.4}{\rmdefault}{\mddefault}{\updefault}{\color[rgb]{0,0,0}Type II critical line}%
}}}}
\put(1565,-1376){\makebox(0,0)[lb]{\smash{{\SetFigFont{10}{12.0}{\rmdefault}{\mddefault}{\updefault}{\color[rgb]{0,0,0}$v_1$}%
}}}}
\put(2206,-1146){\makebox(0,0)[lb]{\smash{{\SetFigFont{10}{12.0}{\rmdefault}{\mddefault}{\updefault}{\color[rgb]{0,0,0}$v_2$}%
}}}}
\put(2226,-1867){\makebox(0,0)[lb]{\smash{{\SetFigFont{10}{12.0}{\rmdefault}{\mddefault}{\updefault}{\color[rgb]{0,0,0}$S$}%
}}}}
\put(3619,-1261){\makebox(0,0)[lb]{\smash{{\SetFigFont{10}{12.0}{\rmdefault}{\mddefault}{\updefault}{\color[rgb]{0,0,0}$L$}%
}}}}
\put(2795,-1388){\makebox(0,0)[lb]{\smash{{\SetFigFont{10}{12.0}{\rmdefault}{\mddefault}{\updefault}{\color[rgb]{0,0,0}$v_3$}%
}}}}
\put(1516,-2432){\makebox(0,0)[lb]{\smash{{\SetFigFont{12}{14.4}{\rmdefault}{\mddefault}{\updefault}{\color[rgb]{0,0,0}Type III critical line}%
}}}}
\put(927,619){\makebox(0,0)[lb]{\smash{{\SetFigFont{10}{12.0}{\rmdefault}{\mddefault}{\updefault}{\color[rgb]{0,0,0}$v_2$}%
}}}}
\put( 46,628){\makebox(0,0)[lb]{\smash{{\SetFigFont{10}{12.0}{\rmdefault}{\mddefault}{\updefault}{\color[rgb]{0,0,0}$v_1$}%
}}}}
\put(957,130){\makebox(0,0)[lb]{\smash{{\SetFigFont{10}{12.0}{\rmdefault}{\mddefault}{\updefault}{\color[rgb]{0,0,0}$S$}%
}}}}
\put(2554,507){\makebox(0,0)[lb]{\smash{{\SetFigFont{10}{12.0}{\rmdefault}{\mddefault}{\updefault}{\color[rgb]{0,0,0}$L$}%
}}}}
\put(266,-362){\makebox(0,0)[lb]{\smash{{\SetFigFont{12}{14.4}{\rmdefault}{\mddefault}{\updefault}{\color[rgb]{0,0,0}Type I critical line}%
}}}}
\put(1788,595){\makebox(0,0)[lb]{\smash{{\SetFigFont{10}{12.0}{\rmdefault}{\mddefault}{\updefault}{\color[rgb]{0,0,0}$v_3$}%
}}}}
\end{picture}%

%% file: cl_typeII_point.tex
\begin{picture}(0,0)%
\includegraphics{cl_typeII_point.pstex}%
\end{picture}%
\setlength{\unitlength}{3947sp}%
\begingroup\makeatletter\ifx\SetFigFont\undefined%
\gdef\SetFigFont#1#2#3#4#5{%
  \reset@font\fontsize{#1}{#2pt}%
  \fontfamily{#3}\fontseries{#4}\fontshape{#5}%
  \selectfont}%
\fi\endgroup%
\begin{picture}(7137,1740)(1083,-2378)
\put(3692,-2063){\makebox(0,0)[lb]{\smash{{\SetFigFont{9}{10.8}{\rmdefault}{\mddefault}{\updefault}{\color[rgb]{0,0,0}$S$}%
}}}}
\put(3692,-2063){\makebox(0,0)[lb]{\smash{{\SetFigFont{9}{10.8}{\rmdefault}{\mddefault}{\updefault}{\color[rgb]{0,0,0}$S$}%
}}}}
\put(8021,-1033){\makebox(0,0)[lb]{\smash{{\SetFigFont{9}{10.8}{\rmdefault}{\mddefault}{\updefault}{\color[rgb]{0,0,0}$L$}%
}}}}
\put(6090,-923){\makebox(0,0)[lb]{\smash{{\SetFigFont{9}{10.8}{\rmdefault}{\mddefault}{\updefault}{\color[rgb]{0,0,0}$d_3$}%
}}}}
\put(7104,-933){\makebox(0,0)[lb]{\smash{{\SetFigFont{9}{10.8}{\rmdefault}{\mddefault}{\updefault}{\color[rgb]{0,0,0}$\alpha_3$}%
}}}}
\put(6526,-950){\makebox(0,0)[lb]{\smash{{\SetFigFont{9}{10.8}{\rmdefault}{\mddefault}{\updefault}{\color[rgb]{0,0,0}$\alpha_3$}%
}}}}
\put(6783,-1126){\makebox(0,0)[lb]{\smash{{\SetFigFont{9}{10.8}{\rmdefault}{\mddefault}{\updefault}{\color[rgb]{0,0,0}$v_3$}%
}}}}
\put(3737,-923){\makebox(0,0)[lb]{\smash{{\SetFigFont{9}{10.8}{\rmdefault}{\mddefault}{\updefault}{\color[rgb]{0,0,0}$v_1$}%
}}}}
\put(4474,-923){\makebox(0,0)[lb]{\smash{{\SetFigFont{9}{10.8}{\rmdefault}{\mddefault}{\updefault}{\color[rgb]{0,0,0}$d_2$}%
}}}}
\put(5151,-923){\makebox(0,0)[lb]{\smash{{\SetFigFont{9}{10.8}{\rmdefault}{\mddefault}{\updefault}{\color[rgb]{0,0,0}$v_2$}%
}}}}
\put(4820,-1110){\makebox(0,0)[lb]{\smash{{\SetFigFont{9}{10.8}{\rmdefault}{\mddefault}{\updefault}{\color[rgb]{0,0,0}$\alpha_2$}%
}}}}
\put(5503,-1105){\makebox(0,0)[lb]{\smash{{\SetFigFont{9}{10.8}{\rmdefault}{\mddefault}{\updefault}{\color[rgb]{0,0,0}$\alpha_2$}%
}}}}
\end{picture}%

%% file: cl_typeII_edge.tex
\begin{picture}(0,0)%
\includegraphics{cl_typeII_edge.pstex}%
\end{picture}%
\setlength{\unitlength}{3947sp}%
\begingroup\makeatletter\ifx\SetFigFont\undefined%
\gdef\SetFigFont#1#2#3#4#5{%
  \reset@font\fontsize{#1}{#2pt}%
  \fontfamily{#3}\fontseries{#4}\fontshape{#5}%
  \selectfont}%
\fi\endgroup%
\begin{picture}(7240,1686)(1355,-2294)
\put(3688,-1995){\makebox(0,0)[lb]{\smash{{\SetFigFont{9}{10.8}{\rmdefault}{\mddefault}{\updefault}{\color[rgb]{0,0,0}$S$}%
}}}}
\put(3688,-1995){\makebox(0,0)[lb]{\smash{{\SetFigFont{9}{10.8}{\rmdefault}{\mddefault}{\updefault}{\color[rgb]{0,0,0}$S$}%
}}}}
\put(3688,-1995){\makebox(0,0)[lb]{\smash{{\SetFigFont{9}{10.8}{\rmdefault}{\mddefault}{\updefault}{\color[rgb]{0,0,0}$S$}%
}}}}
\put(8396,-1016){\makebox(0,0)[lb]{\smash{{\SetFigFont{9}{10.8}{\rmdefault}{\mddefault}{\updefault}{\color[rgb]{0,0,0}$L$}%
}}}}
\put(6465,-906){\makebox(0,0)[lb]{\smash{{\SetFigFont{9}{10.8}{\rmdefault}{\mddefault}{\updefault}{\color[rgb]{0,0,0}$d_3$}%
}}}}
\put(7479,-916){\makebox(0,0)[lb]{\smash{{\SetFigFont{9}{10.8}{\rmdefault}{\mddefault}{\updefault}{\color[rgb]{0,0,0}$\alpha_3$}%
}}}}
\put(6901,-933){\makebox(0,0)[lb]{\smash{{\SetFigFont{9}{10.8}{\rmdefault}{\mddefault}{\updefault}{\color[rgb]{0,0,0}$\alpha_3$}%
}}}}
\put(7158,-1109){\makebox(0,0)[lb]{\smash{{\SetFigFont{9}{10.8}{\rmdefault}{\mddefault}{\updefault}{\color[rgb]{0,0,0}$v_3$}%
}}}}
\put(4849,-906){\makebox(0,0)[lb]{\smash{{\SetFigFont{9}{10.8}{\rmdefault}{\mddefault}{\updefault}{\color[rgb]{0,0,0}$d_2$}%
}}}}
\put(5526,-906){\makebox(0,0)[lb]{\smash{{\SetFigFont{9}{10.8}{\rmdefault}{\mddefault}{\updefault}{\color[rgb]{0,0,0}$v_2$}%
}}}}
\put(5878,-1088){\makebox(0,0)[lb]{\smash{{\SetFigFont{9}{10.8}{\rmdefault}{\mddefault}{\updefault}{\color[rgb]{0,0,0}$\alpha_2$}%
}}}}
\put(5173,-1095){\makebox(0,0)[lb]{\smash{{\SetFigFont{9}{10.8}{\rmdefault}{\mddefault}{\updefault}{\color[rgb]{0,0,0}$\alpha_2$}%
}}}}
\put(3880,-905){\makebox(0,0)[lb]{\smash{{\SetFigFont{9}{10.8}{\rmdefault}{\mddefault}{\updefault}{\color[rgb]{0,0,0}$d_1$}%
}}}}
\put(3534,-905){\makebox(0,0)[lb]{\smash{{\SetFigFont{9}{10.8}{\rmdefault}{\mddefault}{\updefault}{\color[rgb]{0,0,0}$v_1$}%
}}}}
\put(4225,-905){\makebox(0,0)[lb]{\smash{{\SetFigFont{9}{10.8}{\rmdefault}{\mddefault}{\updefault}{\color[rgb]{0,0,0}$v_1'$}%
}}}}
\end{picture}%

%% file: cl_typeIII.tex
\begin{picture}(0,0)%
\includegraphics{cl_typeIII.pstex}%
\end{picture}%
\setlength{\unitlength}{3947sp}%
\begingroup\makeatletter\ifx\SetFigFont\undefined%
\gdef\SetFigFont#1#2#3#4#5{%
  \reset@font\fontsize{#1}{#2pt}%
  \fontfamily{#3}\fontseries{#4}\fontshape{#5}%
  \selectfont}%
\fi\endgroup%
\begin{picture}(7848,1441)(1786,-1876)
\put(6641,-1628){\makebox(0,0)[lb]{\smash{{\SetFigFont{9}{10.8}{\rmdefault}{\mddefault}{\updefault}{\color[rgb]{0,0,0}$d_3$}%
}}}}
\put(5133,-1348){\makebox(0,0)[lb]{\smash{{\SetFigFont{9}{10.8}{\rmdefault}{\mddefault}{\updefault}{\color[rgb]{0,0,0}$\alpha_2$}%
}}}}
\put(2416,-1518){\makebox(0,0)[lb]{\smash{{\SetFigFont{9}{10.8}{\rmdefault}{\mddefault}{\updefault}{\color[rgb]{0,0,0}$\alpha_1$}%
}}}}
\put(3337,-1522){\makebox(0,0)[lb]{\smash{{\SetFigFont{9}{10.8}{\rmdefault}{\mddefault}{\updefault}{\color[rgb]{0,0,0}$\alpha_1$}%
}}}}
\put(5495,-1522){\makebox(0,0)[lb]{\smash{{\SetFigFont{9}{10.8}{\rmdefault}{\mddefault}{\updefault}{\color[rgb]{0,0,0}$v_2$}%
}}}}
\put(5928,-1360){\makebox(0,0)[lb]{\smash{{\SetFigFont{9}{10.8}{\rmdefault}{\mddefault}{\updefault}{\color[rgb]{0,0,0}$\alpha_2$}%
}}}}
\put(2814,-1337){\makebox(0,0)[lb]{\smash{{\SetFigFont{9}{10.8}{\rmdefault}{\mddefault}{\updefault}{\color[rgb]{0,0,0}$v_1$}%
}}}}
\put(3971,-1325){\makebox(0,0)[lb]{\smash{{\SetFigFont{9}{10.8}{\rmdefault}{\mddefault}{\updefault}{\color[rgb]{0,0,0}$d_1$}%
}}}}
\put(8431,-1539){\makebox(0,0)[lb]{\smash{{\SetFigFont{9}{10.8}{\rmdefault}{\mddefault}{\updefault}{\color[rgb]{0,0,0}$\alpha_3$}%
}}}}
\put(7746,-1533){\makebox(0,0)[lb]{\smash{{\SetFigFont{9}{10.8}{\rmdefault}{\mddefault}{\updefault}{\color[rgb]{0,0,0}$\alpha_3$}%
}}}}
\put(7941,-1333){\makebox(0,0)[lb]{\smash{{\SetFigFont{9}{10.8}{\rmdefault}{\mddefault}{\updefault}{\color[rgb]{0,0,0}$v_3$}%
}}}}
\put(9405,-1456){\makebox(0,0)[lb]{\smash{{\SetFigFont{9}{10.8}{\rmdefault}{\mddefault}{\updefault}{\color[rgb]{0,0,0}$L$}%
}}}}
\end{picture}%

%% file: OpaqueCoverSubmit.bbl
\begin{thebibliography}{99}
\bibitem{Akman}
V. Akman, An algorithm for determining an opaque minimal forest of a convex polygon,  Inform. Process. Lett.  24  (1987),  193-198.
\bibitem{Bagemihl}
F. Bagemihl, Some opaque subsets of a square, {Michigan Math.  J.}, {6}
(1959), 99-103.
\bibitem{Brakke}
K. A. Brakke,  The opaque cube problem, The American Math. Monthly, 99 (1992), 866-871.
\bibitem{Brazil}
M. Brazil, D.A. Thomas and J.F. Weng, On the complexity of the Steiner Problem, J. Comb. Optim. 4 (2000), 187-195.
\bibitem{Croft}
H. T. Croft,  Curves intersecting certain sets of great circles on the sphere, J. London Math. Soc. 2 (1969), 461-469.
\bibitem{Du}
D.Z. Du, F.K. Hwang, G.D. Song and G.Y. Ting, Steiner minimal trees on sets of four points, Discrete Comput. Geom. 2 (1987), 401-414.
\bibitem{Dublish}
P. Dublish, An O($n^3$) Algorithm for Finding the Minimal Opaque Forest of a Convex Polygon, { Inform. Process. Lett.}, {29}(1988), 275--276.
\bibitem{Dumitrescu}
A. Dumitrescu, M. Jiang and J. Pach, Opaque sets, Proceedings of the 14th International Workshop on Approximation Algorithms for Combinatorial Optimization Problems, (APPROX 2011), Princeton, New Jersey, LNCS 6845 (2011), 194-205.

\bibitem{Eppstein}
D. Eppstein, Building a better beam detector? \emph{Geometry Junkyard},
\vspace{-5mm}
\begin{verbatim}
http://www.ics.uci.edu/~eppstein/junkyard/beam/ \end{verbatim}

\bibitem{Faber}
V. Faber  and J. Mycielski, The shortest curve that meets all the lines that meet a convex body, Am. Math. Mon. 93 (1986), 796-801.
\bibitem{Faber84}
V. Faber,  J. Mycielski and P. Petersen, On the shortest curve which meets all the lines which meet a circle, Ann. Polon. Math. 44 (1984), 249--256.
\bibitem{Finch}
S. R. Finch, Beam detection constant, Chapter 8.11, pp. 515-519 in Mathematical Constants, Cambridge University Press,  2003.
\bibitem{Hwang}
F. K. Hwang, D. S. Richards, P. Winters, The Steiner Tree Problem (Annals of Discrete Mathematics 53), Elsevier Scinece Publishers B.V., The Netherlands, 1992.
\bibitem{Mazurkiewicz}
S. Mazurkiewicz, Sur un ensemble ferme, punctiforme, qui rencontre toute droite passant par un certain domaine (Polish, French summary), {Prace Mat.-Fiz.}, {
27} (1916), 11-16.
\bibitem{Provan}
J. S. Provan, Convexity and the Steiner tree problem, Networks 18 (1988), 55-72.
\bibitem{Richardson}
T. Richardson and L. Shepp,
The "point" goalie problem,
Discrete Comput. Geom., 30 (2003), 649-669.
\bibitem{Shermer}
Shermer, T., A counterexample to the algorithms for determining opaque minimal forests,  Inform. Process. Lett.  40  (1991),  no. 1, 41--42.
\bibitem{Stewart}
I. Stewart, Cows in the Maze: And Other Mathematical Explorations, Oxford University Press, USA, 2010.
\bibitem{WWZ} D. M. Warme, P. Winter and M. Zachariasen, Exact algorithms for plane Steiner tree problems: a computational study,  pp.\ 81-116, Advances in Steiner Trees, Kluwer Academic Publishers, 2000.
\bibitem{thurston} W. Thurston, sci.math posting, \vspace{-5mm}
\begin{verbatim}
http://www.ics.uci.edu/~eppstein/junkyard/beam/1986.html \end{verbatim}

\end{thebibliography}
